\documentclass[12pt]{amsart}

\usepackage{amssymb}
\usepackage{color}

\usepackage[hmargin=2.5cm,bmargin=2.5cm,tmargin=3cm]{geometry}
\usepackage{graphicx}

\newtheorem{theorem}{Theorem}[section]
\newtheorem{lemma}[theorem]{Lemma}
\newtheorem{corollary}[theorem]{Corollary}

\newtheorem{definition}[theorem]{Definition}

\theoremstyle{remark}
\newtheorem{remark}[theorem]{Remark}
\newtheorem{example}[theorem]{Example}
\newtheorem{exercise}[theorem]{Exercise}

\newcommand{\R}{\mathbb{R}}

\newcommand{\Z}{\mathbb{Z}}
\newcommand{\N}{\mathbb{N}}

\newcommand{\cc}[1]{\overline{#1}}

\newcommand{\term}[1]{\emph{#1}}

\renewcommand{\mod}{\operatorname{\rm mod}}

\DeclareMathOperator{\diag}{diag}
\DeclareMathOperator{\Hess}{Hess}

\usepackage{bm}

\title{An elementary introduction to quantum graphs}

\author{Gregory Berkolaiko}
\address{Department of Mathematics, Texas A\&M University, College
 Station, TX 77843-3368, USA}

\begin{document}

\begin{abstract}
 We describe some basic tools in the spectral theory of Schr\"odinger
 operator on metric graphs (also known as ``quantum graphs'') by
 studying in detail some basic examples.  The exposition is kept as
 elementary and accessible as possible.  In the later sections we
 apply these tools to prove some results on the count of zeros of the
 eigenfunctions of quantum graphs.
\end{abstract}

\maketitle

\section{Introduction}

Studying operators of Schr\"odinger type on metric graphs is a growing
subfield of mathematical physics which is motivated both by direct
applications of the graph models to physical phenomena and by use of
graphs as a simpler setting in which to study complex phenomena
of quantum mechanics, such as Anderson localization, universality of
spectral statistics, nodal statistics, scattering and resonances, to
name but a few.  

The name ``quantum graphs'' is most likely a shortening of the title
of the article ``Quantum Chaos on Graphs'' by Kottos and
Smilansky \cite{KotSmi_prl97}.  The model itself has been studied well
before the name appeared, for example in
\cite{Pau_jcp36,RueSch_jcp53,Rot_crasp83,Bel_laa85,Nic_incol85}.

Several reviews and monographs cover various directions within the
quantum graphs research \cite{GnuSmi_ap06,Post_book12,Mugnolo_book}.
However, when starting a research project with students, both
(post-) graduate and undergraduate, the author felt that a more
elementary introduction would be helpful.  The present manuscript grew
out of the same preparatory lecture repeated, at different points of
time, to several students.  It is basically a collection of minimal
examples of quantum graphs which already exhibit behavior typical to
larger graphs.  We supply the examples with pointers to the more
general facts and theorems.  Only in the last sections we explore a
research topic (the nodal statistics on graphs) in some depth.

For obvious reasons the pointers often lead to the monograph
\cite{BerKuc_graphs}; the notation is kept in line with that book, too.

\section{Schr\"odinger equation on a metric graph}

Consider a graph $\Gamma=(V,E)$, where $V$ is the set of vertices and
$E$ is the set of edges.  Each edge connects a pair of vertices; we
allow more than one edge running between any two vertices.  We also
allow edges connecting vertices to themselves (\term{loops}).  This
freedom creates some notational difficulties, so we ask the reader to
be flexible and forgiving.

Each edge $e$ is assigned a positive length $L_e$ and is thus
identified with an interval $[0, L_e]$ (the direction is chosen
arbitrarily and is irrelevant to the resulting theory).  This makes
$\Gamma$ a \term{metric graph}.  Now a function on a graph is just a
collection of functions defined on individual edges.

The eigenvalue equation for the Schr\"odinger operator is 
\begin{equation}
  \label{eq:eig}
  -\frac{d^2f}{dx^2} + V(x)f(x) = \lambda f(x),
\end{equation}
which is to be satisfied on every edge, in addition to the vertex
matching conditions as follows
\begin{align}
  \label{eq:cont}
  &f(x) \ \mbox{is continuous},\\
  \label{eq:current_cons}
  &\sum_{e \sim v} \frac{df}{dx}(v) = 0.
\end{align}
The continuity means that the values at the vertex agree among all
functions living on the edges attached (or \term{incident}) to the
vertex.  In the second condition (often called \term{current
  conservation condition}), the sum is over all edges attached to the
vertex and the derivative are all taken in the same direction: from
the vertex into the edge.  A looping edge contributes two terms to the
sum, one for each end of the edge.

The function $V(x)$ is called the \term{electric potential} but we
will set it identically to zero in all of the examples below.  Vertex
conditions \eqref{eq:cont}-\eqref{eq:current_cons} are called
\emph{Neumann conditions}\footnote{Other names present in the
  literature are ``Kirchhoff'', ``Neumann-Kirchhoff'', ``standard'',
  ``natural'' etc.}; they can be generalized significantly, but before
we give any more theory, let us consider some examples.

\subsection{Example: a trivial graph --- an interval}

An interval $[0, L]$ is the simplest example of a graph; it has two
vertices (the endpoints of the interval) and one edge.  The continuity
condition is empty at every vertex since there is only one edge.  The
current conservation condition at the vertex $0$ becomes
\begin{equation}
  \label{eq:NC1}
  f'(0) = 0,
\end{equation}
and at the vertex $L$ becomes
\begin{equation}
  \label{eq:NC2}
  -f'(L) = 0.
\end{equation}
The minus sign appeared because we agreed to direct the derivatives
into the edge; of course it is redundant in this particular case.

Let $V(x) \equiv 0$ and consider first the positive eigenvalues,
$\lambda > 0$.  The eigenvalue equation becomes
\begin{equation}
  \label{eq:eig_eq_V0}
  -f'' = k^2 f,
\end{equation}
where for convenience we substituted $\lambda=k^2$.  This is a second
order linear equation with constant coefficients which for $k>0$ is readily
solved by
\begin{equation}
  \label{eq:f_solution}
  f(x) = C_1 \cos(kx) + C_2 \sin(kx).
\end{equation}
Applying the first vertex condition $f'(0)=0$ we get $C_2=0$ and $f(x)
= C_1 \cos(kx)$.  The second vertex condition becomes
\begin{equation}
  \label{eq:f_condition}
  C_1 k \sin(kL) = 0,
\end{equation}
which imposes a condition on $k$ but does nothing to determine $C_1$
(naturally we are not interested in the trivial solution $f(x) \equiv
0$).  We thus get the \emph{eigenvalues} $\lambda = k^2 = \left(\pi
  n/L\right)^2$, $n=1,2,\ldots$ with the corresponding eigenfunctions
$f(x) = \cos\left(\pi n x/L\right)$ defined up to an overall constant
multiplier (as befits eigenvectors and eigenfunctions).

There is one other eigenvalue in the spectrum that we missed:
$\lambda=0$ with the eigenfunction $f(x) \equiv 1$.  While this agrees with
the above formulas with $n=0$, the premise of equation
(\ref{eq:f_solution}) is no longer correct when $\lambda=0$.

\begin{exercise}
  \label{hw:no_neg_eig}
  Solve the eigenvalue equation with $\lambda < 0$ and show that the
  vertex conditions \eqref{eq:NC1} and \eqref{eq:NC2} are never
  satisfied simultaneously (ignore the trivial solution $f(x)\equiv 0$).
\end{exercise}

\begin{exercise}
  \label{hw:nonneg}
  Integrate by parts the scalar product
  \begin{equation}
    \label{eq:scalar_prod}
    \left\langle f, -f'' \right\rangle = \int_0^L \cc{f(x)}
    \left(-f''(x)\right) dx,
  \end{equation}
  to obtain an expression that is obviously non-negative, and thus show
  that it is not necessary to solve \eqref{eq:eig_eq_V0} to conclude
  that there are no negative eigenvalues.
\end{exercise}

We did not try to look for complex eigenvalues.  This is because the
Schr\"odinger operator we defined is self-adjoint (see Thm 1.4.4 of
\cite{BerKuc_graphs}) and therefore has real spectrum.  The spectrum
in the above example is discrete: all eigenvalues are isolated and of
finite multiplicity.  This is true for any graph which is compact (has
finitely many edges, all of which have finite length), see Thm 3.1.1
of \cite{BerKuc_graphs}.  The proof outlined in
Exercise~\ref{hw:nonneg} works for general graphs with Neumann
conditions at all vertices.  The multiplicity of the eigenvalue 0 in
the spectrum can be shown to equal the number of the connected
components of the graph.

\subsection{Example: a trivializable graph with a vertex of degree
  two}
\label{sec:fake_vertex}

Consider now a graph consisting of two consecutive intervals,
$[0,L_1]$ and $[L_1, L_1+L_2]$.  We do not really have to parametrize
the edges starting from 0, so in this example we will employ the
``natural'' parametrization.

Denote the components of eigenfunction living on the two intervals by
$f_1$ and $f_2$ correspondingly.  Solving the equation on the first
edge and enforcing the Neumann condition at $0$ results in $f_1(x) =
C\cos(kx)$.  The conditions at the point $L_1$ are
\begin{align}
  &f_1(L_1) = f_2(L_1), \\
  &-f_1'(L_1) + f_2'(L_1) = 0.
\end{align}
Now, by uniqueness theorem for second order differential equations,
the solution on the second edge is fully determined by its value at
$L_1$ and the value of its derivative.  Thus the solution is still
$f_2(x) = C\cos(kx)$ and there is no change in the solution happening
at $L_1$.  We could have considered the interval $[0,L_1+L_2]$ without
introducing the additional vertex at $L_1$.  This obviously
generalizes to the following rule: \emph{having a Neumann
  vertex of degree 2 is equivalent to having an uninterrupted edge}.

This rule is useful, for example, for when one wants to program a
looping edge but is troubled by the notational difficulties of loops
or multiple edges.  In this case a looping edge can be implemented as
a triangle with two ``dummy'' vertices of degree two.

\subsection{Example: star graph with Neumann endpoints}

\begin{figure}
  \centerline{\includegraphics{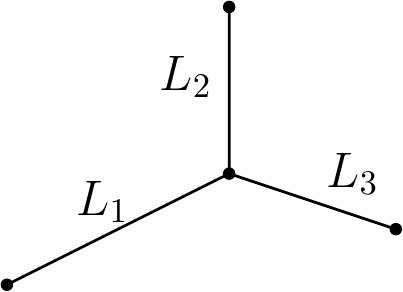}}
  \caption{A star graph with three edges.}
  \label{fig:stargraph}
\end{figure}

Consider now a first non-trivial example: a star graph with 3 edges
meeting at a central vertex, see Fig.~\ref{fig:stargraph}.  Parametrizing the edges
from the endpoints towards the central vertex, we get
\begin{align}
  \label{eq:star_equation}
  &-f_1'' = k^2 f_1, \quad -f_2'' = k^2 f_2, \quad -f_3'' = k^2 f_3,\\
  \label{eq:neumann_per}
  &f_1'(0) = 0, \quad f_2'(0)=0, \quad f_3'(0) = 0, \\
  \label{eq:cont_central}
  &f_1(L_1) = f_2(L_2) = f_3(L_3), \\
  \label{eq:cur_cons_central}
  &-f_1'(L_1) - f_2'(L_2) - f_3'(L_3) = 0,
\end{align}
where in addition to the already familiar equations
\eqref{eq:star_equation} and \eqref{eq:neumann_per} (in three copies),
we have continuity condition at the central vertex in equation
\eqref{eq:cont_central} and current conservation at the central vertex
in equation \eqref{eq:cur_cons_central}.  Note that in equation
\eqref{eq:star_equation}, the eigenvalue $k^2$ is the same on all
three edges.

Equations \eqref{eq:star_equation}-\eqref{eq:neumann_per} are solved
by
\begin{equation}
  \label{eq:soln_legs}
  f_1(x) = A_1\cos(kx), \quad   
  f_2(x) = A_2\cos(kx), \quad 
  f_3(x) = A_3\cos(kx),
\end{equation}
for some constants $A_1$, $A_2$ and $A_3$.  Now the remaining two
equations become, after a minor simplification,
\begin{align}
  \label{eq:cont_central_sub}
  &A_1\cos(kL_1) = A_2\cos(kL_2) = A_3\cos(kL_3), \\
  \label{eq:cur_cons_central_sub}
  &A_1\sin(kL_1) + A_2\sin(kL_2) + A_3\sin(kL_3) = 0.
\end{align}

Dividing equation \eqref{eq:cur_cons_central_sub} by
\eqref{eq:cont_central_sub} cancels the unknown constants, resulting
in
\begin{equation}
  \label{eq:3star}
  \tan(kL_1) + \tan(kL_2) + \tan(kL_3) = 0.
\end{equation}
Squares of the roots $k$ of this equation (which cannot be solved explicitly except
when all $L$s are equal) are the eigenvalues of the star graph.

\begin{exercise}
  \label{hw:lost_roots}
  We ignored the possibility that one or more of the cosines in
  equation \eqref{eq:cont_central_sub} are zero.  Show that the
  more robust (but much longer!) version of equation \eqref{eq:3star}
  is
  \begin{multline}
    \label{eq:3star_robust}
    \sin(kL_1)\cos(kL_2)\cos(kL_3) + \cos(kL_1)\sin(kL_2)\cos(kL_3)\\
    + \cos(kL_1)\cos(kL_2)\sin(kL_3) = 0.
  \end{multline}
  Moreover, the order of the root $k$ of \eqref{eq:3star_robust} is
  equal to the dimension of the corresponding eigenspace.

  For example, if $L_1=L_2=L_3=\pi/2$, the left-hand side of
  \eqref{eq:3star_robust} vanishes at $k=1$ to the second order.  This
  corresponds to two linearly independent solutions,
  \begin{equation}
    \label{eq:efun1}
    (f_1, f_2, f_3) = (\cos(x), -\cos(x), 0) 
   \quad\mbox{and}\quad
   (\cos(x), 0, -\cos(x)). 
  \end{equation}
\end{exercise}

There is actually a lot more that can be (and will be said) about this
simple graph, but we first need to extend the set of possible vertex
conditions that we consider.

\section{Dirichlet condition}

Another possible vertex condition which is compatible with
self-adjointness of the Schr\"odinger operator is the so-called
Dirichlet condition,
\begin{equation}
  \label{eq:Dir_cond}
  f_e(v) = 0 \quad \mbox{for all $e$ incident to $v$}.
\end{equation}
It is usually used only at vertices of degree 1 for the following
reason.  A Dirichlet condition imposed at a vertex of degree 2 or more
fails to relate in any way the individual functions living on the
incident edges.  Thus a graph with a Dirichlet condition at a vertex
$v$ of degree $d_v$ is equivalent to a graph with $v$ substituted with
$d_v$ vertices of degree one, see Fig.~\ref{fig:dirichlet_splits}.

\begin{figure}
  \centering
  \includegraphics{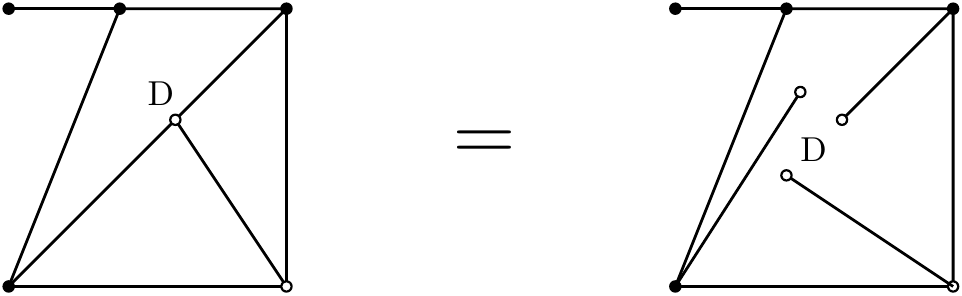}
  \caption{A Dirichlet condition imposed at a vertex of degree $d=3$
    (here and in subsequent figures, the Dirichlet vertices are
    denoted by empty circles) is equivalent to splitting the vertex
    into three and imposing the condition at every new vertex.}
  \label{fig:dirichlet_splits}
\end{figure}

Note that the difference between a Neumann and a Dirichlet
condition at a vertex of degree $d_v$ is minimal: the current
conservation condition is substituted with the condition that
\emph{one} of the function values is equal to zero; the rest is taken
care of by continuity.

\begin{exercise}
  \label{hw:Dir_interval}
  Show that the eigenvalues of the interval $[0,L]$ with Dirichlet
  conditions at both ends are given by $\lambda_n = \left(\pi n /
    L\right)^2$, $n=1, 2, \ldots$ with the eigenfunctions $f^{(n)}(x)
  = \sin\left(\pi n x /L\right)$.
\end{exercise}

\subsection{Example: a star graph with Dirichlet conditions at
  endpoints}
\label{sec:Dstar}

Consider a star graph with $N$ edges.  We parametrize the edges from
the endpoints towards the central vertex, as before.  Solving the
eigenvalue equation $-f'' = k^2f$ and imposing the Dirichlet
condition at $x=0$ results in $f_i(x) = A_i \sin(kx)$, where the
constant $A_i$ depends on the edge.

At the central vertex we have
\begin{align}
  \label{eq:cont_central_subD}
  &A_1\sin(kL_1) = \ldots = A_N\sin(kL_N), \\
  \label{eq:cur_cons_central_subD}
  &A_1\cos(kL_1) + \ldots + A_N\cos(kL_N) = 0.
\end{align}
If we assume that the lengths $L_i$ are incommensurate, we will not
be missing any roots by dividing equation
\eqref{eq:cur_cons_central_subD} by \eqref{eq:cont_central_subD},
leading to the eigenvalue condition
\begin{equation}
  \label{eq:Dirichlet_star}
  \sum_{i=1}^N \cot(kL_i) = 0.
\end{equation}
This condition is very similar to equation~\eqref{eq:3star} we derived
for the star graph with Neumann conditions at the endpoints.  However,
it is now easier to see a connection between the star graph and the
eigenvalue problem of an interval.

\begin{figure}
  \centering
  \includegraphics[scale=1]{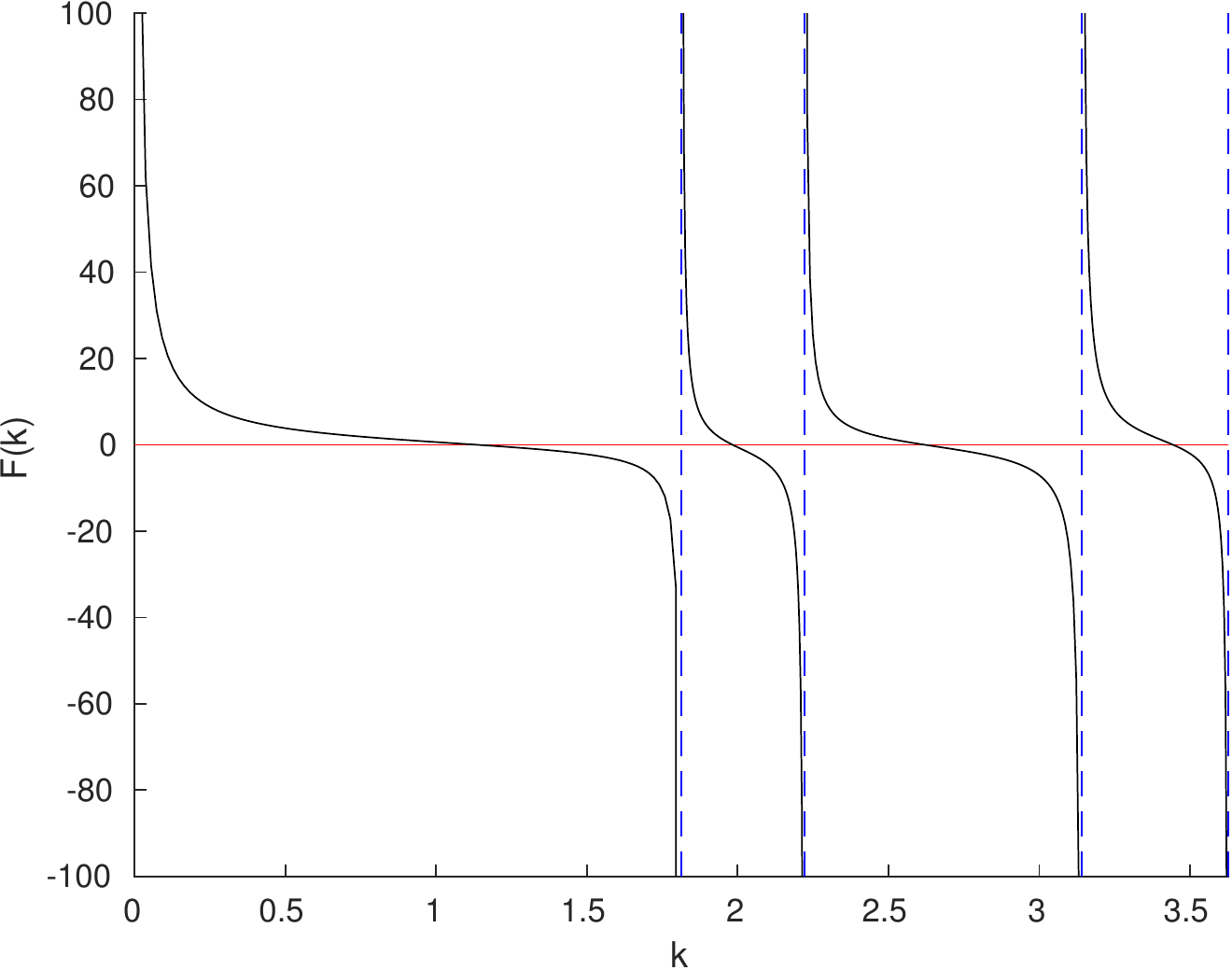}
  \caption{Plot of the right-hand side of (\ref{eq:Dirichlet_star})
    with $3$ edges of lengths $1$, $\sqrt{2}$ and $\sqrt{3}$.  It is a
  function that decreases monotonely between each pair of its poles.
  The poles visible on the plot are $0$, $\pi/\sqrt{3}$,
  $\pi/\sqrt{2}$, $\pi$ and $2\pi/\sqrt{3}$.}
  \label{fig:tan_poles}
\end{figure}

The left-hand side of equation (\ref{eq:Dirichlet_star}) has
derivative of constant sign (negative) except at the poles
$k \in \{n\pi/L_i\}$, $i=1,\ldots,N$, $n\in\Z$. Therefore, between
each pair of consecutive poles there is a single root of equation
(\ref{eq:Dirichlet_star}), see Fig.~\ref{fig:tan_poles}.  The poles
can be interpreted as the square roots of the eigenvalues of the
individual edges of the graph, see Exercise~\ref{hw:Dir_interval}.
Furthermore, the collection of the edges with Dirichlet conditions can
be obtained from the original star graph by changing the central
vertex condition from Neumann to Dirichlet, see
Fig.~\ref{fig:starNtoD}.  As we mentioned already, this can be
effected by changing only one equation in the Neumann conditions,
which is a rank-one perturbation.  To summarize, we found that there
is exactly one eigenvalue of a star graph between any two consecutive
eigenvalues of its rank 1 perturbation. 

\begin{figure}
  \centering
  \includegraphics{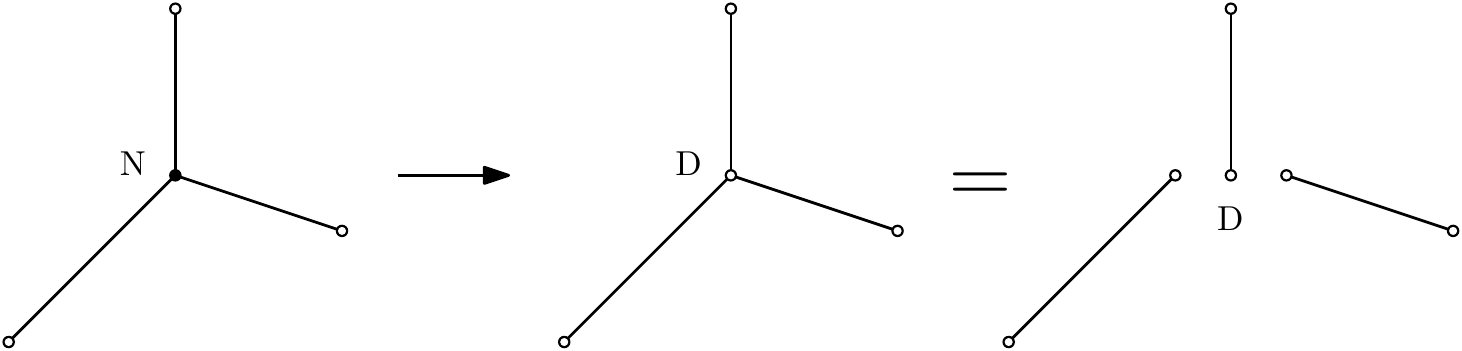}
  \caption{A star graph with the central vertex condition changed to
    Dirichlet splits into a collection of intervals.}
  \label{fig:starNtoD}
\end{figure}

\section{Interlacing inequalities}
\label{sec:interlacing}

Naturally, the observation of Section~\ref{sec:Dstar} applies not only
to star graphs but to any graphs with discrete spectrum.

\begin{lemma}[Neumann--Dirichlet interlacing]
  \label{lem:interlace_ND}
  Let $\Gamma_0$ be a quantum graph with a vertex $v$ which is endowed
  with Neumann conditions.  Let $\Gamma_\infty$ denote the
  graph obtained by changing the conditions at $v$ to Dirichlet.
  Numbering the eigenvalues of both graphs in ascending order starting
  from 1, we have
  \begin{equation}
    \label{eq:interlace_ND}
    \lambda_n(\Gamma_0) \leq \lambda_n(\Gamma_\infty) \leq
    \lambda_{n+1}(\Gamma_0) \leq \lambda_{n+1}(\Gamma_\infty).
  \end{equation}
  An equality between a Neumann and a Dirichlet eigenvalue
  is possible only if the eigenspace of the Neumann
  eigenvalue contains a function vanishing at $v$, or, equivalently,
  the eigenspace of the Dirichlet eigenvalue contains a function
  satisfying the current conservation condition.
\end{lemma}

The proof, which can be found in \cite{BerKuc_incol12,BerKuc_graphs},
uses the standard arguments built upon the minimax characterisation of
eigenvalues of a self-adjoint operator.  It is analogous to the proofs
of Cauchy's Interlacing Theorem or rank-one perturbations for
matrices, see, for example \cite{HornJohnson}.

The following exercise contains an application of
Lemma~\ref{lem:interlace_ND} to the nodal count of eigenfunctions.

\begin{exercise}
  \label{hw:zeros_count}
  Show that if the $n$-th eigenvalue of a star graph with Dirichlet
  endpoints is simple and the corresponding eigenfunction is
  non-vanishing at the central vertex, the eigenfunction has precisely
  $n-1$ zeros in the interior of the graph.

  This statement can be obtained by combining the strict version of
  inequality (\ref{eq:interlace_ND}) with the following observation.
  If $\lambda$ satisfies $\lambda_n(I) \leq \lambda \leq
  \lambda_{n+1}(I)$, where $I$ is the interval $[0,L]$ with Dirichlet
  boundary conditions, then $f(x) = \sin(\sqrt\lambda x)$ has $n$
  zeros on the interval $(0,L)$.
\end{exercise}

A more general version of this statement holds for tree graphs.  This
theorem has a rich history, originally appearing in
\cite{PokPryObe_mz96} and \cite{Sch_wrcm06}; the shortest proof along
the lines outlined in Exercise~\ref{hw:zeros_count} appeared in
\cite{BerKuc_incol12} and in Section 5.2.2 of \cite{BerKuc_graphs}.

\begin{figure}
  \centering
  \includegraphics{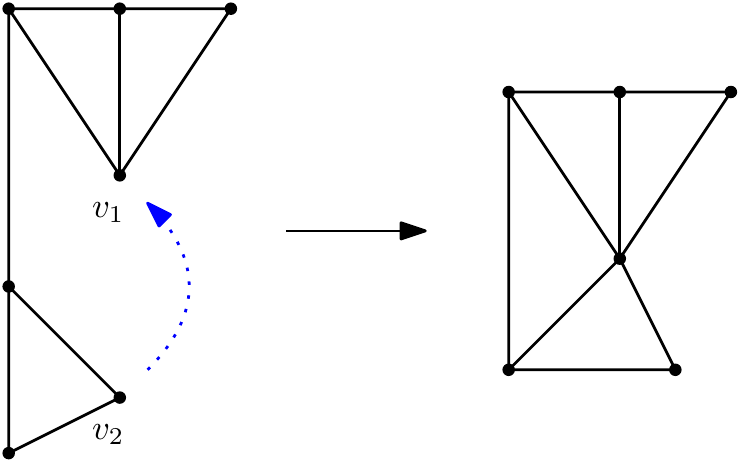}
  \caption{The operation of merging two vertices into one.}
  \label{fig:gluing}
\end{figure}

Another useful interlacing inequality arises when we join two
Neumann vertices to form a single Neumann vertex, see
Fig.~\ref{fig:gluing}.  Since the change in the vertex conditions can
be described as imposing \emph{another} continuity
equality\footnote{This is not entirely correct, as the two current
  conservations conditions are also relaxed into one.  However, in
  terms of quadratic forms, which impose the current conservation
  automatically, the change is indeed a one-dimensional reduction of
  the domain of the form.}, we expect the eigenvalues to increase as a
result.

\begin{lemma}
  \label{lem:interlace_gluing}
  Let $\Gamma$ be a quantum graph (not necessarily connected) with two
  vertices $v_1$ and $v_2$ with Neumann conditions.  Modify
  the graph by merging the two vertices into one, to obtain the graph
  $\Gamma'$. Then
  \begin{equation}
    \label{eq:interlace_gluing}
    \lambda_n(\Gamma) \leq \lambda_n(\Gamma') \leq \lambda_{n+1}(\Gamma).
  \end{equation}
  An equality between an eigenvalue of $\Gamma$ and an eigenvalue of
  $\Gamma'$ is only possible if the eigenspace of $\Gamma$ contains an
  eigenfunction whose values at $v_1$ and $v_2$ are equal or,
  equivalently, the eigenspace of $\Gamma'$ contains an eigenfunction
  which additionally satisfy the current conservation condition with
  respect to the subset $E_{v_1}$ of the edges incident to vertex
  $v_1$ in the graph $\Gamma$.
\end{lemma}

This lemma and Lemma~\ref{lem:interlace_ND} have many applications to
counting zeros of a graph's eigenfunctions, one of which will be
presented in Section~\ref{sec:nodal_dihedral}.  Another application is
to the eigenvalue counting which is the subject of the next section.

\subsection{An application to eigenvalue counting: Weyl's law}

Let us define the eigenvalue \term{counting function} $N_\Gamma(k)$ as the
number of eigenvalues of the graph $\Gamma$ which are smaller than $k^2$,
\begin{equation}
  \label{eq:counting_fn}
  N_\Gamma(k) = \#\left\{ \lambda \in \sigma(\Gamma) 
    : \lambda \leq k^2\right\}.
\end{equation}
This number is guaranteed to be finite since the spectrum of a quantum
graph is discrete and bounded from below (Sec 3.1.1 and Thm 1.4.19 of
\cite{BerKuc_graphs}).  We count the eigenvalues in terms of $k = \sqrt{\lambda}$ as
this is more convenient and can be easily related back to $\lambda$.

The counting function $N_\Gamma(k)$ grows linearly in $k$, with the
slope proportional to the ``size'' of the graph.  This type of result
is known as the ``Weyl's Law''.

\begin{lemma}
  \label{lem:weyl_law}
  Let $\Gamma$ be a graph with Neumann or Dirichlet
  conditions at every vertex.  Then
  \begin{equation}
    \label{eq:weyl_law}
    N(k) = \frac{\mathcal{L}}{\pi} k + O(1),
  \end{equation}
  where $\mathcal{L} = L_1 + \ldots + L_{|E|}$ is the total length of
  the graph's edges and the remainder term is bounded above and below
  by constants independent of $k$.
\end{lemma}

\begin{proof}
  Let us first consider an interval of length $L$ with Dirichlet
  conditions.  We know the eigenvalues are $\left(\pi n / L\right)^2$,
  $n \in \N$, therefore we can express $N_L(k)$ using the integer part
  function,
  \begin{equation*}
    N_L(k) = \left\lfloor \frac{kL}{\pi} \right\rfloor,
  \end{equation*}
  and thus bound it,
  \begin{equation}
    \label{eq:counting_interval}
    \frac{L}{\pi}k - 1 \leq N_L(k) \leq \frac{L}{\pi}k.
  \end{equation}

  Let us now consider the setting of Lemma~\ref{lem:interlace_ND}:
  $\Gamma_0$ is a quantum graph with a vertex $v$ which is endowed
  with Neumann conditions and $\Gamma_\infty$ is the
  graph obtained by changing the conditions at $v$ to Dirichlet.
  Inequality \eqref{eq:interlace_ND} can be rewritten as 
  \begin{equation}
    \label{eq:interlacing_Weyl}
    N_{\Gamma_\infty}(k) \leq N_{\Gamma_0}(k) \leq
    N_{\Gamma_\infty}(k) + 1.
  \end{equation}

  Starting with a graph $\Gamma$, we can change conditions at every
  vertex to Dirichlet.  Applying the interlacing inequality $|V|$
  times (or less, if some vertices are already Dirichlet), we get 
  \begin{equation}
    \label{eq:interlacing_bigD}
    N_{\Gamma_D}(k) \leq N_{\Gamma}(k) \leq
    N_{\Gamma_D}(k) + |V|,
  \end{equation}
  where by $\Gamma_D$ we denote the graph with every vertex conditions
  changed to Dirichlet.  The graph $\Gamma_D$ is just a collection of
  disjoint intervals, each with Dirichlet conditions at the
  endpoints.  The eigenvalue spectrum of $\Gamma_D$ is the union (in
  the sense of multisets) of the spectra of the intervals; the
  counting function is the sum of the interval counting functions.  By
  adding $|E|$ inequalities of type \eqref{eq:counting_interval}, we
  get
  \begin{equation}
    \label{eq:counting_bigD}
    \frac{L_1+\ldots+L_{|E|}}{\pi}k - |E| \leq N_{\Gamma_D}(k) 
    \leq \frac{L_1+\ldots+L_{|E|}}{\pi}k,
  \end{equation}
  leading to the final estimate
  \begin{equation}
    \label{eq:Weyl_result}
    \frac{\mathcal{L}}{\pi}k - |E| \leq N_\Gamma(k) \leq
    \frac{\mathcal{L}}{\pi}k + |V|.
  \end{equation}
\end{proof}

\begin{remark}
  \label{rem:weyl_bound}
  The bounds on the remainder term in the Weyl's law for a graph
  obtained in the proof are of order $|E|$.  However, numerically
  it appears that the counting function follows the Weyl's term much
  more closely.  Getting the optimal bound remains an open question at
  the time of writing.
\end{remark}

\section{Secular determinant}

We will now describe another procedure for deriving an equation for
the eigenvalues of a quantum graph.  Before we describe the general
case, we shall tackle a simple but useful example.

\subsection{Example: lasso (lollipop) graph}

\begin{figure}
  \centering
  \includegraphics{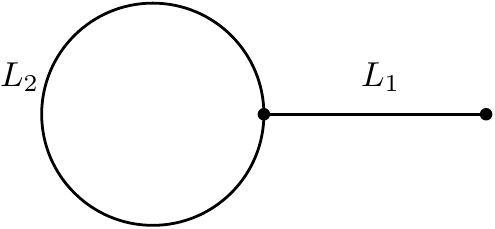}
  \caption{A lasso (or lollipop) graph, consisting of an edge and a loop.}
  \label{fig:lollipop}
\end{figure}

Consider the graph depicted in Fig.~\ref{fig:lollipop}, an edge attached to a loop.
We will impose Neumann conditions at both the attachment point and the
endpoint of the edge.

Let the edge be parametrized by $[0,L_1]$ with 0 corresponding to the
attachment point and the loop be parametrized by $[0,L_2]$.  The
solution of the eigenvalue equation $-f'' = k^2f$ on the edge can be
written as
\begin{equation}
  \label{eq:sol_edge}
  f_1(x) = a_1 e^{ikx} + a_{\bar{1}} e^{ik(L_1-x)},
\end{equation}
valid as long as $k\neq 0$ (we take care of this special case
separately).  The Neumann condition at the endpoint leads to
\begin{equation}
  \label{eq:endpoint_cond_appl}
  f_1'(L_1) = ik a_1 e^{ikL_1} - ik a_{\bar{1}} = 0,
\end{equation}
and therefore
\begin{equation}
  \label{eq:endpoint_cond}
  a_{\bar{1}} = a_1 e^{ikL_1}.
\end{equation}

The solution on the loop we express similarly as
\begin{equation}
  \label{eq:sol_loop}
  f_2(x) = a_2 e^{ikx} + a_{\bar{2}} e^{ik(L_2-x)},
\end{equation}
At the attachment point, the continuity condition reads
\begin{equation}
  \label{eq:attach_cont}
  a_1 + a_{\bar{1}} e^{ikL_1} = a_2 + a_{\bar{2}} e^{ikL_2} = a_{\bar{2}} + a_2 e^{ikL_2},
\end{equation}
while the current conservation is
\begin{equation}
  \label{eq:lasso_cur_cons}
  f_1'(0) + f_2'(0) - f_2'(L_2) = 0,
\end{equation}
which, after simplification, yields
\begin{equation}
  \label{eq:lasso_cur_cons_simpl}
  a_1 - a_{\bar{1}} e^{ik L_1} + a_2 - a_{\bar{2}} e^{ikL_2} +
  a_{\bar{2}} - a_2 e^{ikL_2} = 0.
\end{equation}
Rearranging equations \eqref{eq:attach_cont} and
\eqref{eq:lasso_cur_cons_simpl} we get the system
\begin{align}
  \label{eq:scat1}
  a_1 &= -\frac13 a_{\bar{1}} e^{ik L_1} 
  + \frac23 a_{\bar{2}} e^{ikL_2} + \frac23 a_2 e^{ikL_2}, \\
  \label{eq:scat2}
  a_2 &= \frac23 a_{\bar{1}} e^{ik L_1} 
  + \frac23 a_2 e^{ikL_2} - \frac13 a_{\bar{2}} e^{ikL_2}, \\
  \label{eq:scat3}
  a_{\bar{2}} &= \frac23 a_{\bar{1}} e^{ik L_1} 
  - \frac13 a_2 e^{ikL_2} + \frac23 a_{\bar{2}} e^{ikL_2}.
\end{align}
This system has an interesting ``dynamical'' interpretation, see
Fig.~\ref{fig:scattering}.  Take, for example, the coefficient
$a_{\bar{1}}$ and interpret it as the amplitude of the plain wave
leaving the endpoint vertex in the direction of the loop.  Traversing
the edge (of length $L_1$), it acquires the phase factor of
$e^{ikL_1}$.  Hitting the vertex, it scatters in three directions:
back into the edge with \term{back-scattering} amplitude $-1/3$
contributing to the right-hand side of equation~(\ref{eq:scat1}), and
forward into the two ends of the loop, with \term{forward scattering}
amplitude $2/3$, contributing to equations~(\ref{eq:scat2}) and
(\ref{eq:scat3}).

\begin{figure}
  \centering
  \includegraphics{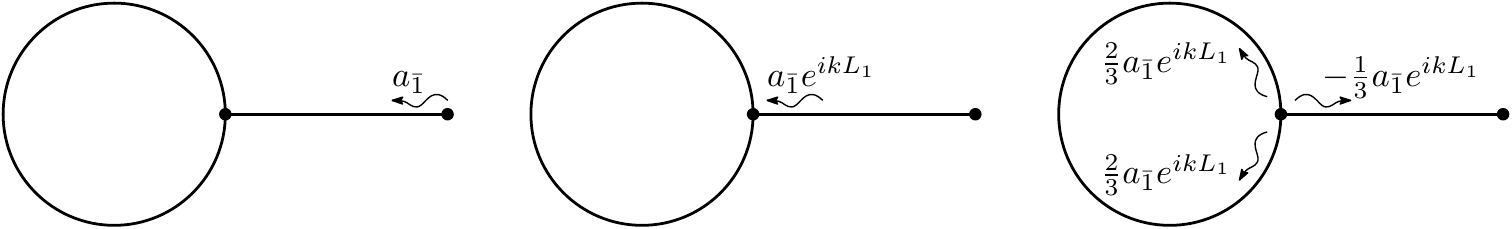}
  \caption{Scattering of a wave on a vertex.  The waves leaves the
    endpoint with amplitude $a_{\bar{1}}$, acquires the phase $e^{ik
        L_1}$ by traversing the edge and scatters in three directions
      on the vertex.}
    \label{fig:scattering}
\end{figure}

The amplitudes $a_1$, $a_2$ and $a_{\bar{2}}$ can be interpreted
similarly and undergo the same process.  Equations
(\ref{eq:endpoint_cond}) and (\ref{eq:scat1})-(\ref{eq:scat3}) can be
interpreted as describing a \term{stationary state} of such a
dynamical process.  They can be written as
\begin{equation}
  \label{eq:stationary}
  \begin{pmatrix}
    0 & -\frac13 & \frac23 & \frac23 \\
    1 & 0 & 0 & 0 \\
    0 & \frac23 & \frac23 & -\frac13 \\
    0 & \frac23 & -\frac13 & \frac23
  \end{pmatrix}
  \begin{pmatrix}
    e^{ikL_1} & 0 & 0 & 0 \\
    0 & e^{ikL_1} & 0 & 0 \\
    0 & 0 & e^{ikL_2} & 0 \\
    0 & 0 & 0 & e^{ikL_2}
  \end{pmatrix}
  \begin{pmatrix}
    a_1 \\ a_{\bar1} \\ a_2 \\ a_{\bar2}
  \end{pmatrix}
  =
  \begin{pmatrix}
    a_1 \\ a_{\bar1} \\ a_2 \\ a_{\bar2}
  \end{pmatrix}.
\end{equation}
Defining two matrices
\begin{equation}
  \label{eq:SandD}
  S =
  \begin{pmatrix}
    0 & -\frac13 & \frac23 & \frac23 \\
    1 & 0 & 0 & 0 \\
    0 & \frac23 & \frac23 & -\frac13 \\
    0 & \frac23 & -\frac13 & \frac23
  \end{pmatrix}
  \qquad\mbox{and}\qquad
  D =
  \begin{pmatrix}
    e^{ikL_1} & 0 & 0 & 0 \\
    0 & e^{ikL_1} & 0 & 0 \\
    0 & 0 & e^{ikL_2} & 0 \\
    0 & 0 & 0 & e^{ikL_2}
  \end{pmatrix},
\end{equation}
we can interpret \eqref{eq:stationary} as saying that the matrix $S
D(k)$ has 1 as its eigenvalue.  Moreover, each eigenvector of the
eigenvalue 1 gives rise to a solution of the eigenvalue equation for
the original differential operator, via equations (\ref{eq:sol_edge})
and (\ref{eq:sol_loop}).  In other words, the geometric multiplicity
of the eigenvalue 1 of the matrix $SD(k)$ is equal to the geometric
multiplicity of the eigenvalue $k^2$ of the differential operator
$-\frac{d^2}{dx^2}$.  Furthermore, both matrices $S$ and $D(k)$ are
unitary, thus for their product $SD(k)$, the algebraic and geometric
multiplicity coincide.  The only point where this relationship can
break down is at $k=0$; this is because solutions (\ref{eq:sol_edge})
and (\ref{eq:sol_loop}) are not valid at this point.

To summarize, we have the following criterion: the value $k^2 \neq 0$
is an eigenvalue of the lasso quantum graph if and only if $k$ is a
root of the equation
\begin{equation}
  \label{eq:sec_eq_lasso}
  \Sigma(k) := \det\left(I - SD(k)\right) = 0.
\end{equation}
The multiplicity of $k^2$ in the spectrum of the graph coincides with
the multiplicity of $k$ as a root of $\Sigma(k)$.  The function
$\Sigma(k)$ is called the \term{secular determinant} of the graph.

To finish this example, we mention that explicit evaluation results in
\begin{equation}
  \label{eq:Sigma_lasso}
  \Sigma(k) = \frac13 (z_2 - 1)(3z_1^2z_2 - z_1^2 + z_2 - 3),
  \qquad \mbox{where }
  z_1 = e^{ikL_1},\ z_2 = e^{ikL_2}.
\end{equation}
We will understand the reason for the factorization of $\Sigma(k)$ in
Section \ref{sec:symmetry}.  Note that the value $k=0$ is a double
root of $\Sigma(k)$ whereas $\lambda=0$ is a simple eigenvalue of the
graph (with constant as the eigenfunction).

\subsection{Secular determinant for a general Neumann graph}
\label{sec:secdet_general}

Let us now consider a general vertex with Neumann conditions
and $d$ edges incident to it.  Writing the solution on $j$-th edge as
\begin{equation}
  \label{eq:sol_jth_edge}
  f_j(x) = a_j e^{ikx} + a_{\bar{j}} e^{ik(L_j-x)},
\end{equation}
we get from the vertex conditions
\begin{align}
  \label{eq:cont_d_edges}
  a_1 + a_{\bar{1}} e^{ikL_1} = \ldots = a_d + a_{\bar{d}} e^{ikL_d}, \\
  \label{eq:cur_cons_d_edges}
  \sum_{j=1}^d a_j - \sum_{j=1}^d  a_{\bar{j}} e^{ikL_j} = 0.
\end{align}
For any $n$, $1\leq n \leq d$, equations~\eqref{eq:cont_d_edges} imply
\begin{equation}
  \sum_{j=1}^d a_j + \sum_{j=1}^d  a_{\bar{j}} e^{ikL_j} 
  = d \left( a_n + a_{\bar{n}} e^{ikL_n} \right).
\end{equation}
Subtracting from this equation \eqref{eq:cur_cons_d_edges} and solving
for $a_n$ results in
\begin{equation}
  \label{eq:scat_NK}
  a_n = -a_{\bar{n}} e^{ikL_n} + \frac2d \sum_{j=1}^d
  a_{\bar{j}} e^{ikL_j},
\end{equation}
which is a generalization of both \eqref{eq:endpoint_cond} and
\eqref{eq:scat1}-\eqref{eq:scat3}, with $d=1$ and $d=3$ correspondingly.

Now, it is clear how to generalize the matrices $S$ and $D(k)$ in equation \eqref{eq:SandD}.  Every
edge of the graph gives rise to two directed edges which inherit the
length of the edge.  The two directed edges corresponding to the same
undirected edge are \term{reversals} of each other.  The reversal of a
directed edge $j$ is denoted by $\bar{j}$.  Consider the
$2|E|$-dimensional complex space, with dimensions indexed by the directed
edges.  The matrix $D(k)$ is diagonal with entries
\begin{equation}
  \label{eq:Dk_entries}
  D(k)_{j,j} = e^{ik L_j},
\end{equation}
while the matrix $S$ has the entries
\begin{equation}
  \label{eq:S_entries}
  S_{j', j} =
  \begin{cases}
    \frac2{d_v} - 1, & \mbox{if }j' = \bar{j}, \\
    \frac2{d_v}, & \mbox{if $j'$ follows $j$ and }j' \neq \bar{j}, \\
    0, & \mbox{otherwise}.
  \end{cases}
\end{equation}
The edge $j'$ \term{follows} $j$ if the end-vertex of $j$ is the start
vertex of $j'$; $d_v$ denotes the degree of the end-vertex of $j$.
The matrix $S$ is sometimes called the \term{bond scattering matrix}.

\begin{exercise}
  Prove that the matrix $S$ defined by \eqref{eq:S_entries} on a graph
  is unitary.
\end{exercise}

As before, every eigenvector of $SD(k)$ with $k\neq0$ corresponds to
an eigenfunction of the graph via equation \eqref{eq:sol_jth_edge}.
We therefore have the following theorem.

\begin{theorem}
  \label{thm:sec_det_NK}
  Consider a graph with Neumann conditions at every vertex.
  The value $\lambda = k^2 \neq 0$ is an eigenvalue of the operator
  $-d^2/dx^2$ if and only if $k$ is the solution of
  \begin{equation}
    \label{eq:sec_det_NK}
    \Sigma(k) := \det\left(I - SD(k)\right) = 0,
  \end{equation}
  where $S$ and $D(k)$ are defined in equations \eqref{eq:Dk_entries}
  and \eqref{eq:S_entries}.
\end{theorem}

\begin{exercise}
  \label{hw:dirichlet_vertex}
  Incorporate vertices of degree 1 with Dirichlet conditions by
  showing that back-scattering from such a vertex has amplitude $-1$
  (in contrast with $1$ from a Neumann vertex, see
  the first case of \eqref{eq:S_entries} with $d=1$).  This matches
  with the classic reflection principle of the wave equation on an
  interval.
\end{exercise}

\begin{remark}
  Another method used to prove the Weyl's Law
  (Lemma~\ref{lem:weyl_law}) is to observe \cite{Gri_pams13} that the
  eigenvalues of the matrix $SD(k)$ lie on the unit circle and move in
  the counter-clockwise direction as $k$ is increased.  There are
  $2|E|$ of the eigenvalues and their average angular speed can be
  calculated to be $\mathcal{L}/|E|$.  Thus the frequency of an
  eigenvalue crossing the positive real axis is the average speed
  times the number of eigenvalues divided by the length of the circle,
  giving $2 |E| \times \mathcal{L} / |E| / (2\pi) = \mathcal{L} /
  \pi$.
\end{remark}

\subsection{Example: secular determinant of star graphs with three edges}
\label{sec:sec_det_star}

Consider again a star graph with 3 edges and Neumann
conditions everywhere.  Ordering the directed edges as $\left[1, 2, 3,
  \bar{1}, \bar{2}, \bar{3}\right]$, where $j$ is directed away from the
central vertex and $\bar{j}$ is directed towards the central vertex,
the bond scattering matrix is
\begin{equation}
  \label{eq:scat_matr_Nstar}
  S =
  \begin{pmatrix}
    0 & 0 & 0 & -1/3 & 2/3 & 2/3 \\
    0 & 0 & 0 & 2/3 & -1/3 & 2/3 \\
    0 & 0 & 0 & 2/3 & 2/3 & -1/3 \\
    1 & 0 & 0 & 0 & 0 & 0 \\
    0 & 1 & 0 & 0 & 0 & 0 \\
    0 & 0 & 1 & 0 & 0 & 0
  \end{pmatrix}.
\end{equation}
Evaluating the secular determinant $\Sigma(k)$ from equation
\eqref{eq:sec_det_NK}, we get
\begin{equation}
  \label{eq:Sigma_Nstar}
  \Sigma(k) = -z_1^2 z_2^2 z_3^2 
  - \frac13 \left(z_1^2 z_2^2 + z_2^2 z_3^2 + z_3^2 z_1^2\right)
  + \frac13 \left(z_1^2 + z_2^2 + z_3^2\right) + 1, 
  \qquad \mbox{where } z_j = e^{ikL_j}.
\end{equation}
Dividing this expression by $i z_1z_2z_3 = i e^{ik(L_1+L_2+L_3)}$ and
using the Euler's formula we can transform equation~\eqref{eq:Sigma_Nstar} to
the form~\eqref{eq:3star_robust} obtained previously (up to a constant).

According to Exercise~\ref{hw:dirichlet_vertex}, supplying the star
graph with Dirichlet conditions at the endpoints results in the bond
scattering matrix
\begin{equation}
  \label{eq:scat_matr_Dstar}
  S =
  \begin{pmatrix}
    0 & 0 & 0 & -1/3 & 2/3 & 2/3 \\
    0 & 0 & 0 & 2/3 & -1/3 & 2/3 \\
    0 & 0 & 0 & 2/3 & 2/3 & -1/3 \\
    -1 & 0 & 0 & 0 & 0 & 0 \\
    0 & -1 & 0 & 0 & 0 & 0 \\
    0 & 0 & -1 & 0 & 0 & 0
  \end{pmatrix},
\end{equation}
and the secular determinant
\begin{equation}
  \label{eq:Sigma_Dstar}
  \Sigma(k) = z_1^2 z_2^2 z_3^2 
  - \frac13 \left(z_1^2 z_2^2 + z_2^2 z_3^2 + z_3^2 z_1^2\right)
  - \frac13 \left(z_1^2 + z_2^2 + z_3^2\right) + 1.
\end{equation}
Again, dividing it by $z_1z_2z_3 = e^{ik(L_1+L_2+L_3)}$
brings it close to the previously obtained form,
equation~\eqref{eq:Dirichlet_star}, namely to
\begin{multline}
  \label{eq:robust_Dstar}
  e^{-ik(L_1+L_2+L_3)}\Sigma(k) = \sin(kL_1)\sin(kL_2)\cos(kL_3)\\ 
  + \cos(kL_1)\sin(kL_2)\sin(kL_3) + \sin(kL_1)\cos(kL_2)\sin(kL_3).
\end{multline}

\subsection{Real secular determinant}

Theorem~\ref{thm:sec_det_NK} gives a handy tool for looking for the
eigenvalues of a quantum graph.  However, $\Sigma(k)$, as defined by
equation~\eqref{eq:sec_det_NK} is a complex valued function (that
needs to be evaluated on the real line --- at least when looking for
positive eigenvalues).  A complex function is equal to zero when both
real and imaginary part are equal to zero.  It would be nicer to have
one equation instead of two.

There are several indications that it should be possible.  First, the fact that
we do have roots of a complex function on the real line\footnote{Since
  our operator is self-adjoint, the eigenvalues $\lambda=k^2$ must be
  real and therefore the roots of $\Sigma(k)$ are restricted to real
  and imaginary axes.  Since the operator is bounded from below and
  unbounded from above, infinitely many of the eigenvalues must lie on
  the real axis} is atypical; it suggests that the function has some
symmetries.  Second, in the two examples that we considered in
Section~\ref{sec:sec_det_star} we succeeded in making the secular
determinant real.  It turns out that the same method works in general.

\begin{lemma}
  \label{lem:zeta_real_valued}
  Let 
  \begin{equation}
    \label{eq:tot_length}
    \mathcal{L} = \sum_{e\in E} L_e
  \end{equation}
  denote the total length of the graph.
  The analytic function 
  \begin{equation}
    \label{eq:sec_func}
    \zeta(k) = \frac{e^{-ik \mathcal{L}}}{\sqrt{\det S}} \det\left(I - S D(k) \right),
  \end{equation}
  is real on the real line and has the same zeros as the secular
  determinant $\Sigma(k)$.
\end{lemma}

\begin{proof}
  We remark that
  \begin{equation}
    \label{eq:det_Dk}
    \det(D(k)) = e^{2ik \mathcal{L}},
  \end{equation}
  and the matrix $D(k)$ is unitary for real $k$.  Denoting the unitary
  matrix $SD(k)$ by $U$, and using the identity $\det(AB) = \det A\,
  \det B$, we can rewrite
  \begin{equation}
    \label{eq:zeta_rewrite}
    \zeta(k) = (\det U)^{-1/2} \det\left(I - U \right),
  \end{equation}
  Using the unitarity of $U$, we now evaluate
  \begin{align*}
    \overline{\zeta(k)} &= (\det U)^{1/2}\,
    \det\left(I - U^*\right)\\
    &= (\det U)^{1/2}\,
    \det\left(U - I \right)\,
    \det U^* = \zeta(k),
  \end{align*}
  where we used the identities
  \begin{equation}
    \label{eq:matrix_identities}
    I - U^* = UU^* - U^* = (U-I)U^*
   \quad \mbox{and} \quad
    \det\left(U^*\right) = \left( \det U \right)^{-1}.
  \end{equation}
\end{proof}

\begin{exercise}
  \label{hw:detS}
  Show that for a graph with Neumann conditions at every
  vertex except for $n$ vertices of degree 1 where Dirichlet
  conditions are imposed, the determinant of $S$ is $\det S =
  (-1)^{|E|-|V|+n}$. 
\end{exercise}

A detailed study of the secular determinant, including the
interpretation of the coefficients of the polynomials like
(\ref{eq:Sigma_Nstar}) and (\ref{eq:Sigma_Dstar}) in terms of special
closed paths on the graph, appears in \cite{BanHarJoy_jpa12}.

\subsection{Remarks on numerical calculation of graph eigenvalues}

Expression \eqref{eq:sec_func} which is guaranteed to be real for real
$k$ allows for a simple way to compute eigenvalues of a quantum
graph: find roots of a real function.

The naive method is to evaluate the function $\zeta(k)$ on a dense
enough set of points to get bounds for the eigenvalues, within which
the bisection method can be employed.  The bisection method can be
substituted by a more sophisticated tool such as Brent-Dekker
method.  For a small graph, the function $\zeta(k)$, which is a
combination of trigonometric functions of incommensurate frequency,
can be derived explicitly, giving access to its derivative and
therefore Newton-like methods.

However, the initial evaluation may miss a pair of almost-degenerate
eigenvalues.  To check for this possibility, it is very effective to
plot the difference between the counting function for the computed
eigenvalues and the Weyl's Law.  The approximate location of the
missed eigenvalues (if any) can be seen very clearly, see
Fig.~\ref{fig:counting} for a typical example.  See also
\cite{Sch_incol06} for another method.

\begin{figure}
  \centering
  \includegraphics{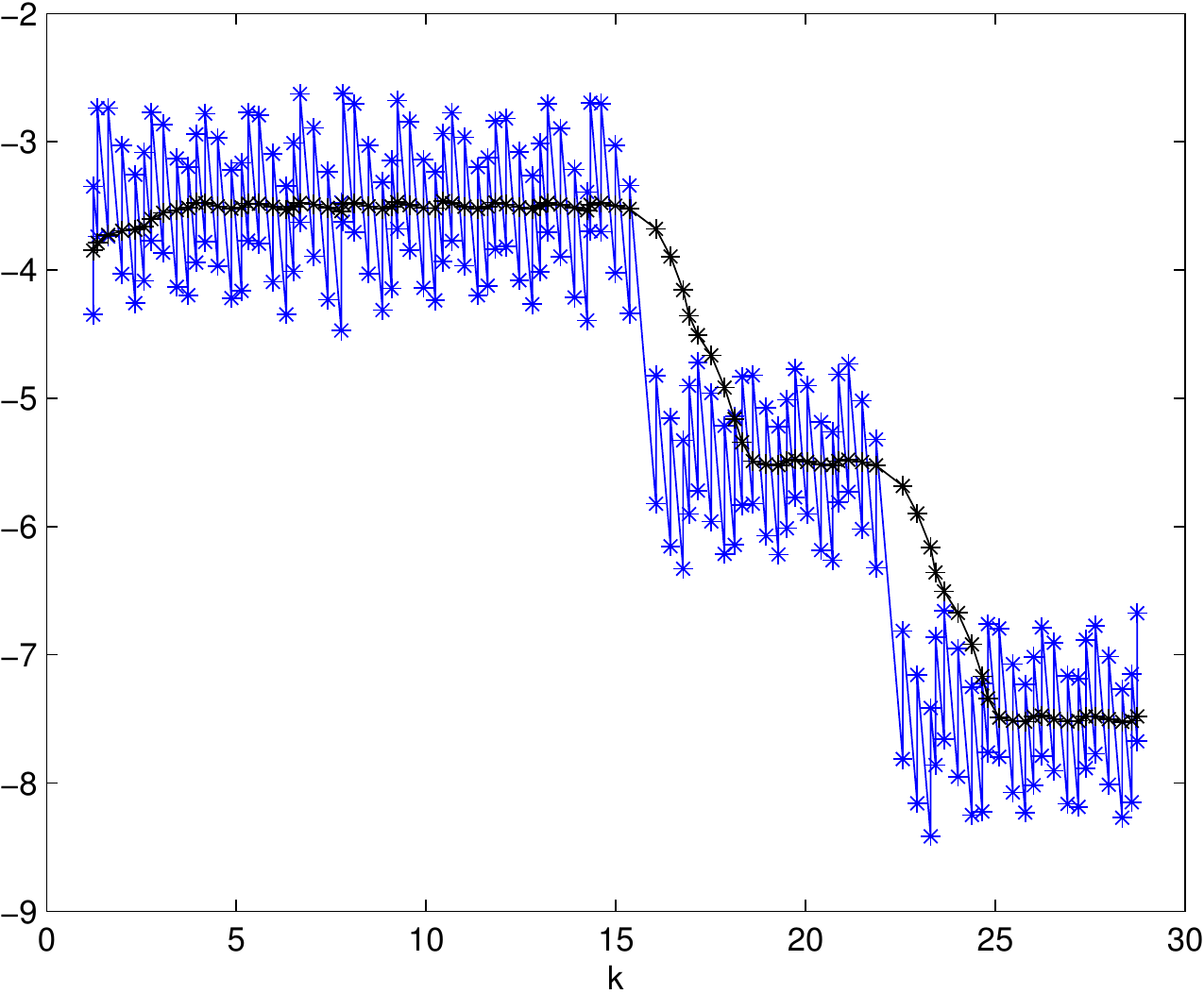}
  \caption{The difference between the counting function of the numerically found eigenvalue and the Weyl's estimate.  The plot shows that two pairs of eigenvalues were missed, around $k=15.5$ and $k=22$.}
  \label{fig:counting}
\end{figure}

A smarter method for the initial step is to use the interlacing
inequalities of Lemma~\ref{lem:interlace_ND} to bracket the
eigenvalues.  Unfortunately, this involves computing eigenvalues of
another graph, but it may be much simpler, as in the case of the star
graphs.

\section{Symmetry and isospectrality}
\label{sec:symmetry}

\subsection{Example: 3-mandarin graph}
\label{sec:mandarin_fact}

The 3-mandarin graph is a graph with two vertices and three edges
connecting them, see Fig.~\ref{fig:mandarin}.  If one is uncomfortable
with multiple edges running between a pair of vertices, extra
Neumann vertices of degree two may be placed on some edges,
see Section~\ref{sec:fake_vertex}.

\begin{figure}
  \centering
  \includegraphics{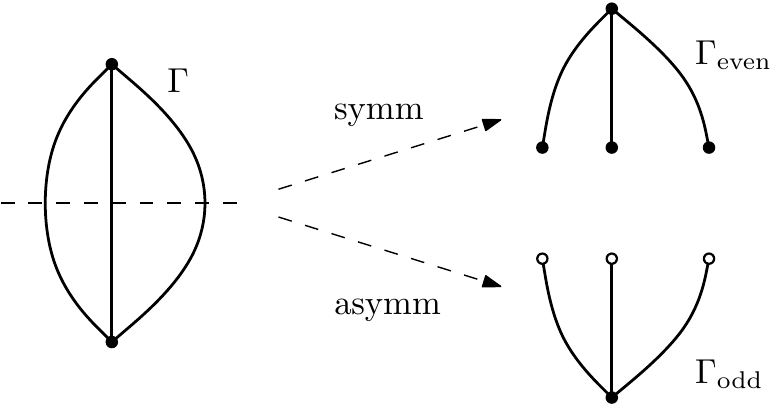}
  \caption{The mandarin graph with 3 edges (left) and its
    decomposition into the even and odd quotients (right). The
    Dirichlet vertices are distinguished as empty cicrles.}
  \label{fig:mandarin}
\end{figure}

Labelling the edges running down by $1$, $2$ and $3$, and using the
ordering $\left[1, 2, 3, \bar1, \bar2, \bar3\right]$, the matrix $S$
becomes
\begin{equation}
  \label{eq:scat_matr_mandarin}
  S =
  \begin{pmatrix}
    0 & 0 & 0 & -1/3 & 2/3 & 2/3 \\
    0 & 0 & 0 & 2/3 & -1/3 & 2/3 \\
    0 & 0 & 0 & 2/3 & 2/3 & -1/3 \\
    -1/3 & 2/3 & 2/3 & 0 & 0 & 0 \\
    2/3 & -1/3 & 2/3 & 0 & 0 & 0 \\
    2/3 & 2/3 & -1/3 & 0 & 0 & 0
  \end{pmatrix}.
\end{equation}
Denoting, as before, $z_j = e^{i k L_j}$, $j=1,2,3$, we get $D(k) =
\diag(z_1, z_2, z_3, z_1, z_2, z_3)$.  The secular determinant
simplifies to
\begin{multline}
  \label{eq:Sigma_mandarin}
  \Sigma(k) = \left( -z_1 z_2 z_3 
    - \frac13 \left(z_1 z_2 + z_2 z_3 + z_3 z_1\right)
    + \frac13 \left(z_1 + z_2 + z_3\right) + 1 \right) \\
  \times \left( z_1 z_2 z_3 
    - \frac13 \left(z_1 z_2 + z_2 z_3 + z_3 z_1\right)
    - \frac13 \left(z_1 + z_2 + z_3\right) + 1 \right),
\end{multline}
where the factors coincide with the secular determinants we obtained
for the star graphs with three edges and Neumann and Dirichlet conditions
at the endpoints, correspondingly, see equations
\eqref{eq:Sigma_Nstar} and \eqref{eq:Sigma_Dstar}, modulo the change
$z_j^2 \leftrightarrow z_j$.

The reason for this factorization is the symmetry.  The graph, as
shown in Fig.~\ref{fig:mandarin}, is symmetric with respect to the
vertical (up-down) reflection.  This means that the reflected
eigenfunction is still an eigenfunction.  Denote by $Rf$ the reflected
version of an eigenfunction $f$.  By linearity, $f_e = f+Rf$ and $f_o
= f-Rf$ satisfy the eigenvalue equation with the same $\lambda$.  They
may be identically zero, but not both at the same time, since
$(f_e+f_o)/2 = f \not\equiv 0$.  And under the action of $R$, $f_e$ is even
and $f_o$ is odd:
\begin{equation}
  \label{eq:even_odd}
  Rf_e = f_e, \qquad Rf_o = -f_o.
\end{equation}
Using this idea one can show that every eigenspace has a basis of
eigenfunctions each of which is either even or odd.  Indeed, starting
with an arbitrary basis of size $m$, we produce $2m$ even/odd
combinations from them.  These combinations span the eigenspace (since
$f = (f_e+f_o)/2$), so it remains to choose $m$ linearly independent
vectors among them.

Consider an odd eigenfunction on the mandarin graph.  Let $x_m$ be the
midpoint of the first edge.  This point is fixed under the action
of reflection, thus $(Rf)(x_m) = f(x_m)$.  On the other hand, $Rf =
-f$, therefore $f(x_m) = -f(x_m) = 0$.  The same applies for the
midpoint of every edge.  Therefore every odd eigenfunction is the
eigenfunction of the half of the mandarin with Dirichlet conditions at
the midpoints.  The converse is also true: starting from an
eigenfunction of the half with Dirichlet boundary, we obtain an
eigenfunction of the full graph by planting two copies and multiplying
one of them by $-1$.

A similar reasoning for even eigenfunctions of the mandarin graph
shows that they are in one-to-one correspondence with the
eigenfunctions of the half of the graph with Neumann boundary.  The
half of a 3-mandarin is a star graph with three edges, see
Fig.~\ref{fig:mandarin}; the edge lengths are half of the mandarin's.

It turns out that this symmetry of the mandarin graph (and the
corresponding factorization \eqref{eq:Sigma_mandarin}) leads to
interesting anomalies in the number of of zeros of the eigenfunctions.
This subject will be visited again in
Section~\ref{sec:nodal_mandarin}.

\subsection{Quotient graphs}

To put the observations of the previous section on a more formal
footing, the Hilbert space $H(\Gamma)$ of functions on the mandarin
graph (that are sufficiently smooth and satisfy correct vertex
conditions) can be decomposed into the direct sum of two orthogonal
subspaces $H_o$ and $H_e$, which are invariant with respect to
operator $-d^2/dx^2$ acting on the graph.  Restrictions of the
operator to these subspaces can be identified with this operator
acting on two smaller graphs, a star with Dirichlet ends and a star
with Neumann ends.

Such a smaller graph, together with its vertex conditions, is called
\term{quotient graph} and was introduced by Band, Parzanchevski and
Ben-Shach \cite{BanParBen_jpa09,ParBan_jga10} to study isospectrality.
To produce a quotient graph one needs a quantum graph, a group of
symmetries (not necessarily the largest possible) and a representation
of this group.  We will not describe the full procedure, which can be
learned from the already mentioned papers and the forthcoming article
\cite{BanBerJoyLiu_prep16}.  Instead we will briefly describe its
consequences and explain one particular construction that leads to a pair
of graphs with identical spectra (i.e.\ an \emph{isospectral} pair).

Let $\Gamma$ be a graph with a finite group of symmetries $G$.  To
each irreducible representation $\rho$ there corresponds a subspace
$H_\rho$ of the Hilbert space $H(\Gamma)$.  This is the subspace of
functions that transform according to the representation $\rho$ when
acted upon by the symmetries from $G$.  In some sources, such
functions are called \term{equivariant vectors}; the subspace $H_\rho$
is called the \term{isotypic component}.  The subspaces corresponding
to different irreducible representations are orthogonal, the space
$H(\Gamma)$ is a direct sum of $H_\rho$ over all irreps of the group
$G$.

If $\rho$ has dimension $d>1$, then every eigenvalue of the
Hamiltonian restricted to $H_\rho$ has multiplicity which is a
multiple of $d$.  Moreover, the secular determinant $\Sigma(k)$
factorizes into factors that correspond to the irreps $\rho$ of $G$.
Each factor is raised to a power which is the dimension of the
corresponding $\rho$ (hence the degeneracy of the corresponding
eigenvalue),
\begin{equation}
  \label{eq:fact_Sigma}
  \Sigma_\Gamma(k) = \prod_{\text{irreps }\rho}
  \Big(\Sigma_\rho(k)\Big)^{\dim(\rho)}.
\end{equation}

\begin{figure}
  \centering
  \includegraphics{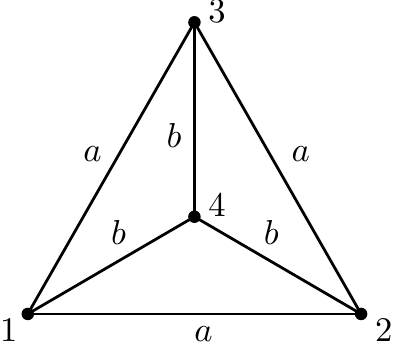}
  \caption{A tetrahedron graph which is invariant under rotation by
    $2\pi/3$ and horizontal reflection.  Its group of symmetry is
    $S_3$: an arbitrary permutation of vertices $1$, $2$ and $3$.}
  \label{fig:tetrahedron}
\end{figure}

\begin{example}[from \cite{BanBerJoyLiu_prep16}]
  Consider the tetrahedron graph from Fig.~\ref{fig:tetrahedron}.
  This graph is symmetric under reflection and rotation by $2\pi/3$.
  The corresponding group of symmetries is $S_3$, the group of
  permutations of 3 objects; in this case it can be thought of as
  permuting the vertices number $1$, $2$ and $3$.  The group $S_3$ has
  three irreducible representations, \term{trivial},
  \term{alternating} and $2d$.  The first two are one-dimensional,
  while the latter is $2$-dimensional (as suggested by its name).  The
  secular determinant $\Sigma(k)$ has the corresponding factorization
  \begin{equation}
    \label{eq:fact_S3}
    \Sigma(k) = \frac1{27} (z_a - 1) (3z_a z_b^2 - z_b^2 + z_a - 3) 
    (3 z_a^2 z_b^2 + 2 z_a z_b^2 - z_a^2 + z_b^2 - 2 z_a - 3)^2.
  \end{equation}
\end{example}

Note that it may happen that for some $\rho$ the subspace $H_\rho$ is
trivial.  In this case the corresponding factor $\Sigma_\rho(k)$ is $1$
and there are no eigenvalues corresponding to this representation.

\subsection{Example: dihedral graphs}

The present example originates from \cite{BanShaSmi_jpa06} and is the
origin of the theory of \cite{BanParBen_jpa09,ParBan_jga10}.  Consider the graph
on Fig.~\ref{fig:predihedral}(a).  It has the symmetries of the square, thus its full group
of symmetries is the dihedral group of degree four $D_4$ (and order
eight, hence another notation, $D_8$, which causes much confusion).

\begin{figure}
  \centering
  \includegraphics{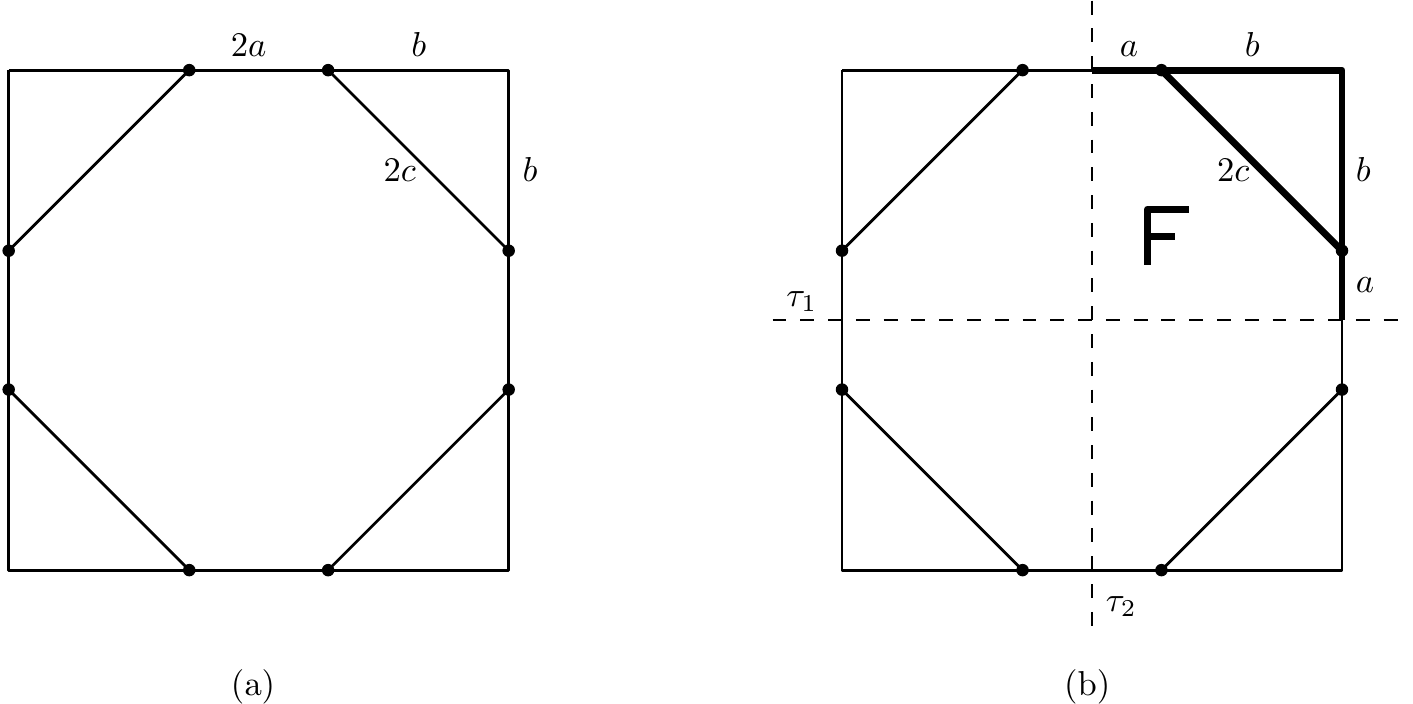}
  \caption{(a) A graph with symmetry group $D_4$, the dihedral group
    of degree 4; (b) symmetry axes of the first subgroup and a choice
    of fundamental domain.}
  \label{fig:predihedral}
\end{figure}

We will first consider the subgroup generated by the vertical (up-down)
reflection $\tau_1$ and the horizontal (left-right) reflection
$\tau_2$.  One of the irreducible representations of this subgroup is
\begin{equation}
  \label{eq:irrep_R1}
  \tau_1 \mapsto (1), \qquad \tau_2 \mapsto (-1),
\end{equation}
where $(1)$ (correspondingly $(-1)$) stands for the operation of
multiplication by $1$ (correspondingly $-1$).

\begin{figure}
  \centering
  \includegraphics{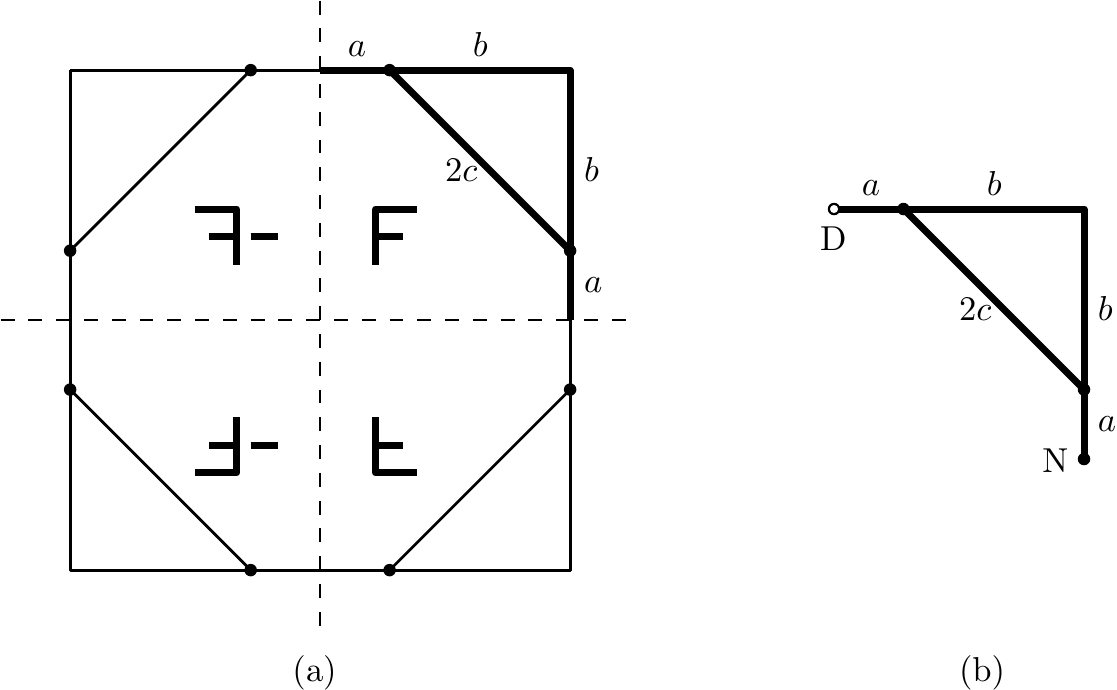}
  \caption{Constructing a quotient graph with respect to the
    representation \eqref{eq:irrep_R1}.}
  \label{fig:dihedral1}
\end{figure}

To understand the functions that transform according to representation
\eqref{eq:irrep_R1}, we choose as the \term{fundamental domain} (a
subgraph that covers the entire graph under the action of the
subgroup) the top right quarter of the graph shown in
Fig.~\ref{fig:predihedral}(b), and plant there a function $F$.
Applying the vertical reflection $\tau_1$, we find that in the bottom
right quarter of the graph, the function must be equal to $F$
multiplied by $1$ and suitably reflected.  Applying the horizontal
reflection $\tau_2$, we find that the left side of the graph must be
populated by the copies of the function $F$ multiplied by $-1$, see
Fig.~\ref{fig:dihedral1}(a).

We have already discovered in Section~\ref{sec:mandarin_fact} that at
the point where $F$ meets $-F$, the function must vanish (i.e.\ have
the Dirichlet condition), whereas at the point where $F$ meets $F$ the
condition must be Neumann.  Thus we obtain the quotient graph in
Fig.~\ref{fig:dihedral1}(c).

Let us now repeat the same procedure but choose the diagonal reflections
$d_1$ and $d_2$ as the generators of the subgroup, together with the
representation
\begin{equation}
  \label{eq:irrep_R2}
  d_1 \mapsto (1), \qquad d_2 \mapsto (-1).
\end{equation}
Starting with a fundamental domain, reflecting and
multiplying it as prescribed by \eqref{eq:irrep_R2}, we fill the
entire graph as in Fig.~\ref{fig:dihedral2}(a).  The corresponding conditions on the
fundamental domain are shown in Fig.~\ref{fig:dihedral2}(b).

\begin{figure}
  \centering
  \includegraphics{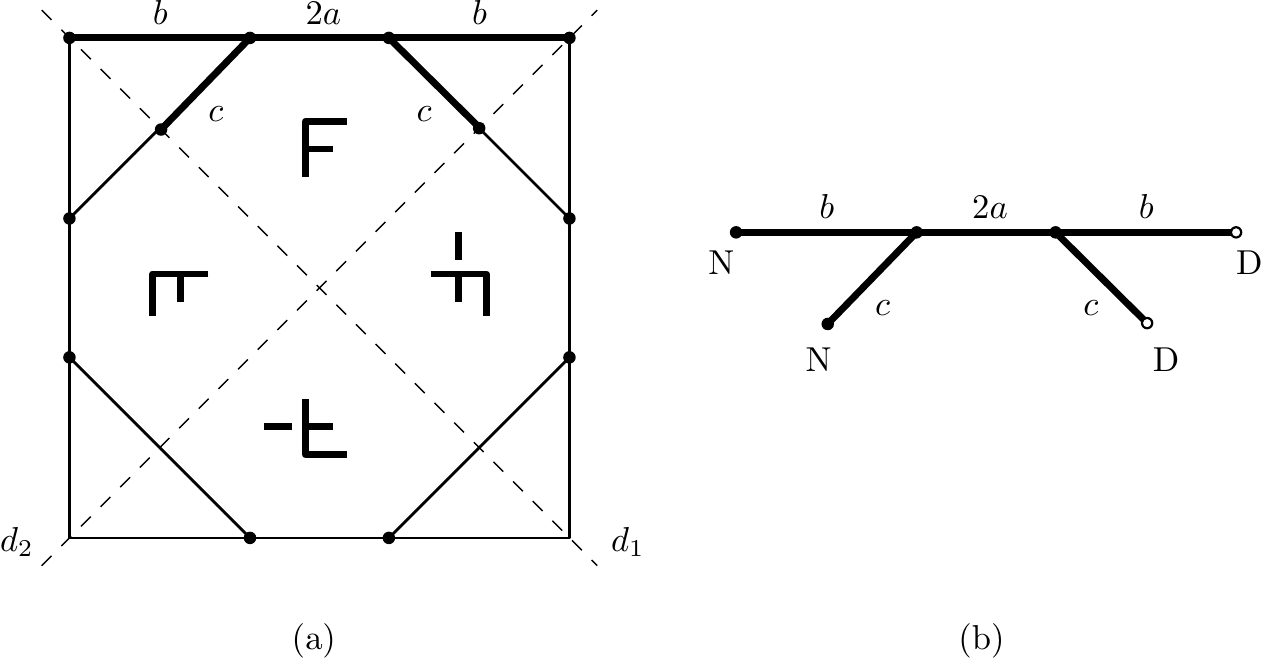}
  \caption{Constructing a quotient graph with respect to the
    representation \eqref{eq:irrep_R2}.}
  \label{fig:dihedral2}
\end{figure}

The most interesting feature of the two quotient subgraphs is that they
are \term{isospectral}, i.e. have exactly the same eigenvalues.  This
can be shown by a transplantation procedure {\it a la\/} Buser
\cite{Bus_aif86,Bus+_imrn94}, which describes a unitary equivalence of
the corresponding operators.  Another possibility (admittedly more
tedious) is to find the secular determinants of the two graphs.

\begin{figure}
  \centering
  \includegraphics{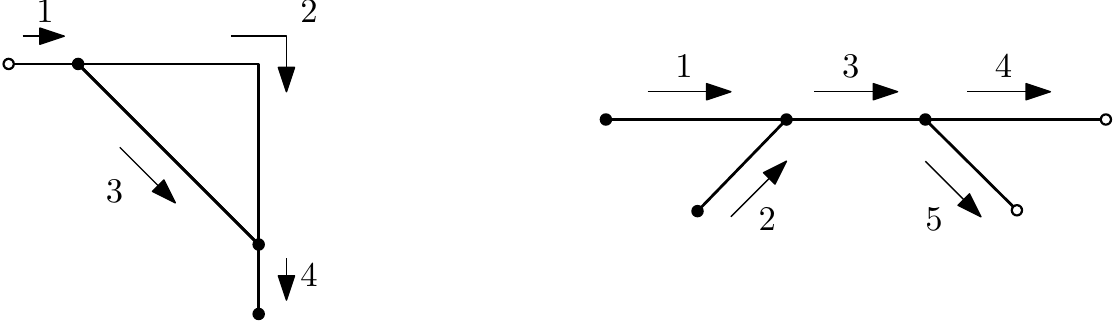}
  \caption{Numbering of edges of the two isospectral graph.}
  \label{fig:dihedrals_numbering}
\end{figure}

To this end, number the edges of the two graphs as shown in
Fig.~\ref{fig:dihedrals_numbering}.  Starting with the graph with a
cycle (which we will now call the \term{dihedral graph}, and its
partner the \term{dihedral tree}), order its edges as $[1,2,3,4,\bar1,
\bar2, \bar3, \bar4]$.  Then
\begin{equation}
  \label{eq:SD_dihedral}
  S  =
  \begin{pmatrix}
    0 & 0 & 0 & 0 & -1 & 0 & 0 & 0 \\
    2/3 & 0 & 0 & 0 & 0 & -1/3 & 2/3 & 0\\
    2/3 & 0 & 0 & 0 & 0 & 2/3 & -1/3 & 0\\
    0 & 2/3 & 2/3 & 0 & 0 & 0 & 0 & -1/3 \\
    -1/3 & 0 & 0 & 0 & 0 & 2/3 & 2/3 & 0 \\
    0 & -1/3 & 2/3 & 0 & 0 & 0 & 0 & 2/3 \\
    0 & 2/3 & -1/3 & 0 & 0 & 0 & 0 & 2/3 \\
    0 & 0 & 0 & 1 & 0 & 0 & 0 & 0
  \end{pmatrix},
\end{equation}
and $D(k) = \diag(z_a, z_b^2, z_c^2, z_a, z_a, z_b^2, z_c^2, z_a)$,
where $z_s=e^{iks}$, $s=a,b,c$.  Evaluating the determinant, we obtain
\begin{equation}
  \label{eq:Sigma_dihedral}
  \Sigma(k) = -z_a^4 z_b^4 z_c^4 + \frac19 \left(z_a^4 z_b^4 + z_b^4
    z_c^4 + z_c^4 z_a^4\right) + \frac89 \left(z_a^4 - 1\right)z_b^2 z_c^2
    -\frac19 \left(z_a^4  + z_b^4 + z_c^4\right) + 1.
\end{equation}

Repeating the procedure for the dihedral tree, which has
$10\times10$ matrices $S$ and $D(k)$, we arrive to the secular
determinant which is \emph{again given by expression}
\eqref{eq:Sigma_dihedral}.  It is easy to check that $0$ is eigenvalue
of neither graph and therefore the graphs are isospectral.

The glimpse of the underlying reason for the isospectrality can be
seen in the secular determinant of the original graph (that of
Fig.~\ref{fig:predihedral}(a)).  As predicted by \eqref{eq:fact_Sigma}, it factorizes:
\begin{align*}
  \label{eq:Sigma_orig}
  \Sigma(k) = &
  \left(9 z_a^4 z_b^4 z_c^4 - \left(z_a^4 z_b^4 + z_b^4
    z_c^4 + z_c^4 z_a^4\right) - 8\left(z_a^4 - 1\right)z_b^2 z_c^2
    + z_a^4  + z_b^4 + z_c^4 - 1\right)^2 \\
  &\times \left(3 z_a^2 z_b^2 z_c^2 + z_a^2 z_b^2 + z_b^2
    z_c^2 + z_c^2 z_a^2  - z_a^2 - z_b^2 - z_c^2 - 3\right) \\
  &\times \left(3 z_a^2 z_b^2 z_c^2 + z_a^2 z_b^2 - z_b^2
    z_c^2 + z_c^2 z_a^2  - z_a^2 + z_b^2 + z_c^2 + 3\right) \\
  &\times \left(3 z_a^2 z_b^2 z_c^2 - z_a^2 z_b^2 + z_b^2
    z_c^2 - z_c^2 z_a^2  - z_a^2 + z_b^2 + z_c^2 - 3\right) \\
  &\times \left(3 z_a^2 z_b^2 z_c^2 - z_a^2 z_b^2 - z_b^2
    z_c^2 - z_c^2 z_a^2  - z_a^2 - z_b^2 - z_c^2 + 3\right),
\end{align*}
up to an overall factor.  The last four terms correspond to the four
one-dimensional representations of the group $D_4$, while the first one,
squared, corresponds to the two-dimensional representation.  The term
inside the square also coincides with the secular determinant of the
two quotient graphs.  This suggests that although we constructed them
as quotients by the (one-dimensional) representations of two different
symmetry subgroups, they are both realizable as quotients by the
two-dimensional representation of the whole group.  This is indeed
shown in \cite{BanParBen_jpa09}, together with a general criterion for
isospectrality involving induction of representations from subgroups to
the whole group.

\section{Magnetic Schr\"odinger operator and nodal count}

Magnetic field is introduced into the Schr\"odinger equation via the
magnetic \term{vector potential} usually denoted $A(x)$.  In our case,
$A(x)$ is a one-dimensional vector: it changes sign if the direction
of the edge is reversed.  The Schr\"odinger eigenvalue equation then
takes the form
\begin{equation}
  \label{eq:mag_schrod}
  -\left( \frac d{dx} - i A(x)\right)^2 f(x) + V(x) f(x) = k^2 f(x),
\end{equation}
where the square is interpreted in the sense of operators, i.e.
\begin{multline}
  \label{eq:square_op}
  \left( \frac d{dx} - i A(x)\right)^2 f(x) = \left( \frac d{dx} 
    - iA(x)\right) \Big(f'(x) - i A(x)f(x)\Big)\\ 
  = f''(x) - i \Big(A(x)f(x)\Big)' - iA(x) f'(x) - A^2(x)f(x). 
\end{multline}

To understand the ``strange'' definition of $A(x)$ a little better,
consider the equation
\begin{equation}
  \label{eq:const_A}
  -\left( \frac d{dx} - i A \right)^2 f(x) = k^2 f(x),
\end{equation}
where $A$ is a ``normal'' constant, on the interval $[0,L]$.
Solutions of this equation, $e^{\pm ikx + iAx}$, under the change of
variables $x \mapsto L-x$ become solutions of a slightly different
equation, 
\begin{equation}
  \label{eq:const_A_swap}
  -\left( \frac d{dx} + i A \right)^2 f(L-x) = k^2 f(L-x).
\end{equation}
But this variable change is just a reparametrization of the interval
in terms of the distance from the other end and should not affect the
laws of physics.  Letting $A$ to be the ``one-form'' which transforms
according to $A(L-x) = -A(x)$ addresses this problem.

To understand the effect of the magnetic potential $A(x)$ on the
secular determinant derived in Section~\ref{sec:secdet_general}, we
write the solution of \eqref{eq:mag_schrod} with $V\equiv0$ on the
$j$-th edge in the form
\begin{equation}
  \label{eq:mag_sol_jth_edge_int}
  f_j(x) = a_j e^{ikx + i\int_0^x A(x)} + a_{\bar{j}} e^{ik(L_j-x) +
    i\int_{L_j}^xA(x)}.
\end{equation}
Taking $A(x)$ to be the constant vector on the edge (as we will see,
this results in no loss of generality), we obtain a somewhat more
manageable form
\begin{equation}
  \label{eq:mag_sol_jth_edge}
  f_j(x) = a_j e^{i(k+A)x} + a_{\bar{j}} e^{i(k-A)(L_j-x)},
\end{equation}
which also highlights the fact that the wave travelling in the negative
direction (the second term) feels the magnetic potential as $-A$.

The vertex conditions also change with the introduction of the
magnetic field.  It is convenient to define the operator
\begin{equation}
  \label{eq:Dmag_operator}
  D = \frac d{dx} - i A(x),
\end{equation}
so that the Schr\"odinger operator can be written as $-D^2 + V$ and
the Neumann vertex conditions become
\begin{align}
  \label{eq:cont_mag}
  &f(x) \ \mbox{is continuous},\\
  \label{eq:current_cons_mag}
  &\sum_{e \sim v} Df(v) = 0,
\end{align}
which is to be compared to \eqref{eq:cont}-\eqref{eq:current_cons}.

Applying these conditions to the solution form
\eqref{eq:mag_sol_jth_edge}, we get
\begin{align}
  \label{eq:cont_d_edges_mag}
  a_1 + a_{\bar{j}} e^{ikL_1 + iA_{\bar1} L_1} = \ldots = a_d + a_{\bar{d}}
  e^{ikL_d + iA_{\bar{d}}L_d}, \\
  \label{eq:cur_cons_d_edges_mag}
  \sum_{j=1}^d a_j - \sum_{j=1}^d  a_{\bar{j}} e^{ikL_j + iA_{\bar{j}}L_j} = 0,
\end{align}
which assumes that all edges attached to $v$ are oriented outward and
also introduces the notation $A_{\bar{j}} = -A_j$ in agreement with
the nature of $A$.  In fact the same answer would be obtained had we
started with equation \eqref{eq:mag_sol_jth_edge_int} instead, and defined
\begin{equation}
  \label{eq:Aj_def}
  A_j = \frac1{L_j} \int_0^{L_j} A(x), 
  \qquad A_{\bar{j}} = \frac1{L_j} \int_{L_j}^0 A(x).
\end{equation}
Either way, the only change in the definition of the secular
determinant $\Sigma(k) = \det(I - SD(k))$ is in the diagonal matrix
$D(k)$ which becomes
\begin{equation}
  \label{eq:Dk_def_mag}
  D(k)_{b,b} = e^{ikL_b + i \int_b A(x)}.
\end{equation}

As we just saw, the precise nature of $A(x)$ is not important; the
only quantity that enters $\Sigma(k)$ is the integral of $A(x)$.  In
fact, even more is true: only the values of the integral of $A(x)$
around the \emph{cycles} of the graph are important.

\begin{definition}
  The \term{flux} of the magnetic field given by $A(x)$ through an
  oriented cycle $\gamma$ of $\Gamma$ is the integral of $A(x)$ over
  the cycle,
  \begin{equation}
    \label{eq:flux_def}
    \alpha_\gamma = \int_\gamma A(x).
  \end{equation}
\end{definition}

\begin{theorem}
  Consider two operators on the same quantum graph $\Gamma$ differing
  only in their magnetic potentials $A_1(x)$ and $A_2(x)$. Then they
  are unitarily equivalent if their fluxes through every cycle on
  $\Gamma$ are equal modulo $2\pi$.  In fact, it is enough to compute
  the fluxes through a fundamental set of $\beta = |E|-|V|+1$ cycles.
  A magnetic perturbation of a Schr\"odinger operator on a graph
  $\Gamma$ is thus fully specified by a set of $\beta$ numbers between
  $-\pi$ and $\pi$. 
\end{theorem}

The proof of this simple theorem can be found in, for example,
\cite{KosSch_cmp03} or \cite{BerKuc_graphs}, Section 2.6.


\subsection{Example: Dihedral graph}

To study the influence of the magnetic perturbation on the eigenvalues
of the dihedral graph, we put the magnetic flux $\alpha$ through its
loop by using $D(k) = \diag(z_a,\, e^{i\alpha} z_b^2,\, z_c^2, z_a, z_a,\,
e^{-i\alpha} z_b^2,\, z_c^2, z_a)$ together with the matrix $S$ as given
in (\ref{eq:SD_dihedral}).  This results in the secular determinant
\begin{equation}
  \label{eq:Sigma_dihedral_mag}
  \Sigma(k) = -z_a^4 z_b^4 z_c^4 + \frac19 \left(z_a^4 z_b^4 + z_b^4
    z_c^4 + z_c^4 z_a^4\right) + \frac89 \cos(\alpha) \left(z_a^4 - 1\right)z_b^2 z_c^2
    -\frac19 \left(z_a^4  + z_b^4 + z_c^4\right) + 1.
\end{equation}

\begin{figure}
  \centering
  \includegraphics{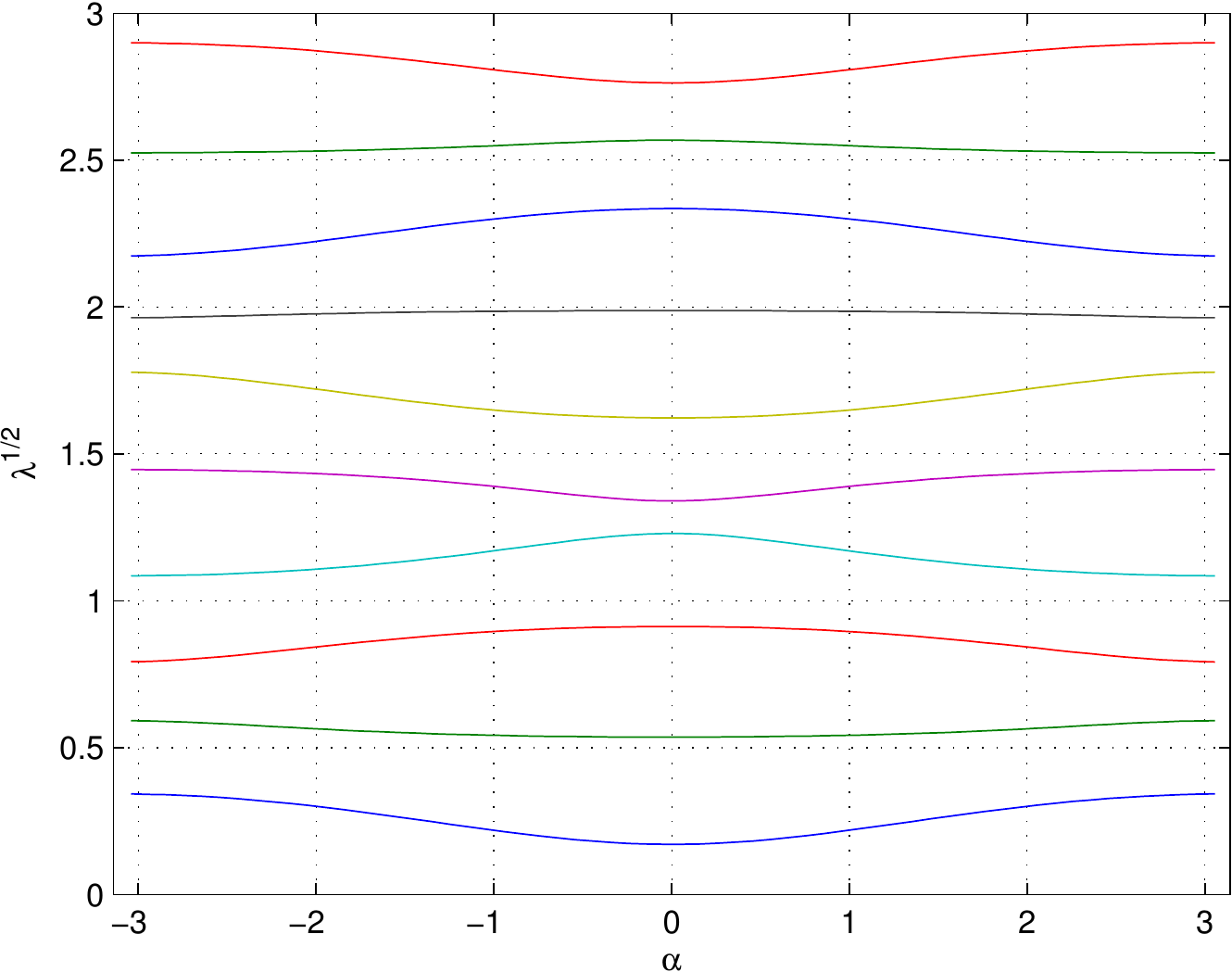}
  \caption{The square roots of the first few eigenvalues of the
    dihedral graph as function of the magnetic flux through the cycle.
    The lengths are $a=\pi$, $b=1$, $c=\sqrt{2}$.}
  \label{fig:mag_eig}
\end{figure}

We plot the first few eigenvalues $\lambda_n(\alpha)$ of the resulting
graph as a function of the magnetic flux $\alpha$,
Fig.~\ref{fig:mag_eig}.  Since the flux is only important modulo
$2\pi$, we plot the eigenvalues over the interval $[-\pi,\pi]$.  We
observe that the eigenvalues are symmetric (even) functions with
respect to $\alpha=0$: this can be seen directly by complex
conjugating equation~(\ref{eq:mag_schrod}): if $f(x)$ is an
eigenfunction with potential $A(x)$ then $\overline{f(x)}$ is an
eigenfunction with potential $-A(x)$ with the same eigenvalue.
Sometimes the eigenvalue has a minimum at $\alpha=0$ and sometimes a
maximum.  These events do not alternate as in the Hill's
equation\footnote{\emph{Hill's differential equation}
  \cite{MagnusWinkl_hills} is a Schr\"odinger equation on $\R^1$ with
  periodic potential; its Floquet reduction is equivalent to a quantum
  graph in the shape of a circle with a magnetic flux $\alpha$ through
  it.  It is an important result of Hill's equation theory that the
  minima and maxima of $\lambda_n(\alpha)$ at $\alpha=0$ alternate
  with $n$.}; there is no strict periodicity there, but we will be able
to extract some information about them.

\subsection{Magnetic--nodal connection}

\begin{figure}
  \centering
  \includegraphics[scale=0.6]{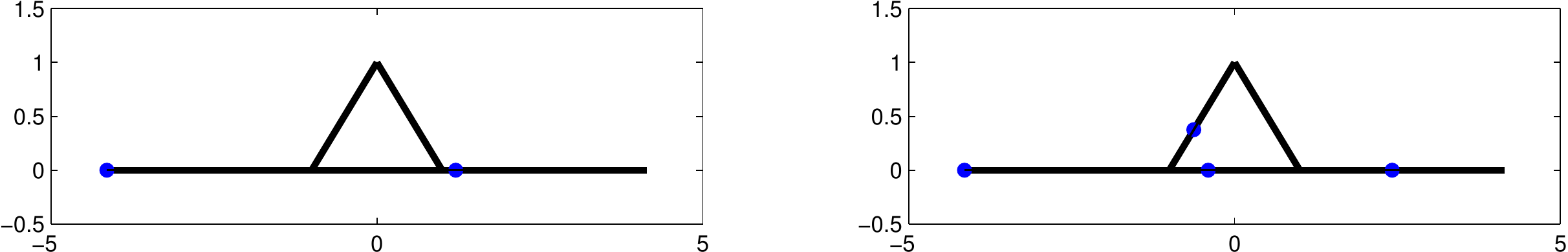}
  \caption{Location of zeros of two eigenfunctions of the dihedral
    graph.  Eigenfunctions number $n=2$ and $n=3$ are displayed.  The
    zero at the left endpoint is due to the Dirichlet condition; such
    zeros are not counted when we report the nodal count of an
    eigenfunction.}
  \label{fig:dihedral_ef}
\end{figure}

Let us for a moment come back to the case of no magnetic field,
$\alpha=0$, and study the eigenfunctions of the dihedral graph, see
Fig.~\ref{fig:dihedral_ef}.  In Table~\ref{tab:eig_dih} we give
the results of a numerical calculation for one choice of the graph's
lengths.  We list the sequential number of the eigenvalue, starting
from 1, its value, the number of zeros of the corresponding
eigenfunction, and, somewhat arbitrarily, a description of the
behavior of the eigenvalue of the magnetic dihedral graph (see
Fig.~\ref{fig:mag_eig}), namely whether it has a maximum or a minimum
at $\alpha=0$.

\begin{table}[ht]
  \centering
  \begin{tabular}{l|cccccccccc}
    $n$ &1 & 2 & 3 & 4 & 5 & 6 & 7 & 8 & 9 \\
    \hline
    $\lambda_n(0)$ & 0.1708 & 0.5359 & 0.9126 & 1.2294 & 1.3398 & 1.6225 &
    1.9877 & 2.3349 & 2.5680 \\         
    $\#$ zeros & 0 & 1 & 3 & 4 & 4 & 5 & 7 & 8 & 9 \\
    $\lambda_n(\alpha)$ at 0 & min & min & max & max & min & min & max
    & max & max
  \end{tabular}
  \vspace{0.3cm}
  \caption{Eigenvalues of a dihedral graph with the lengths $a=\pi$,
    $b=1$ and $c=\sqrt{2}$,  the number of zeros of the
    corresponding eigenfunction, and the behavior of the magnetic
    eigenvalue around the point of no magnetic field.}
  \label{tab:eig_dih}
\end{table}

A careful reader will observe a curious pattern: the magnetic
eigenvalue $\lambda_n(\alpha)$ appears to have a minimum whenever the
number of zeros of the eigenfunction, which we will denote $\phi_n$ is
less than $n$ and a maximum whenever $\phi_n = n$.  The connection
between the two was discovered in \cite{Ber_apde13} for discrete
Laplacians, an alternative proof given in \cite{Col_apde13} and an
extension to quantum graphs proved in \cite{BerWey_ptrsa14}.

Before we formulate this result we need to recall some definitions
from multivariate calculus.

\begin{definition}
  Let $F(x_1, \ldots, x_\beta)$ be a twice differentiable function of
  $\beta$ variables.
  \begin{itemize}
  \item The point $x^* = (x_1^*, \ldots, x_\beta^*)$ is a \term{critical
      point} of the function $F$ if all first derivatives of $F$
    vanish at $x^*$,
    \begin{equation*}
      \label{eq:critical_def}
      \frac{\partial F}{\partial x_j}(x_1^*, \ldots, x_\beta^*) = 0.
    \end{equation*}
  \item The \term{Hessian} matrix of $F$ at $x^*$ is the matrix of all
    second derivatives of $F$,
    \begin{equation*}
      \label{eq:Hessian_def}
      \Hess(F) = \left( \frac{\partial^2 F}{\partial x_j \partial x_k} \right)_{j,k=1}^\beta.
    \end{equation*}
    The matrix $\Hess(F)$ is symmetric therefore all of its
    eigenvalues are real.
    \item A critical point is called \term{non-degenerate} if the
      Hessian evaluated at the point is a non-degenerate matrix (has
      no zero eigenvalues).
    \item The \term{Morse index} of a critical point is the number of
      the negative eigenvalues of its Hessian.  The Second Derivative
      Test says that if the Morse index of a non-degenerate critical
      point is zero, the point is a local minimum; if it equals the
      dimension of the space, the point is a maximum.
  \end{itemize}
\end{definition}

\begin{theorem}[Berkolaiko--Weyand \cite{BerWey_ptrsa14}]
  \label{thm:main_mag}
  Let $\Gamma^{\boldsymbol\alpha}$ be a quantum graph with magnetic
  Schr\"o\-dinger operator characterized by magnetic fluxes
  $\boldsymbol\alpha = (\alpha_1, \ldots, \alpha_\beta)$ through a
  fundamental set of cycles.  Let $\psi$ be the eigenfunction of
  $\Gamma^0$ that corresponds to a simple eigenvalue
  $\lambda = \lambda_n(\Gamma^0)$ of the graph with zero magnetic
  field.  We assume that $\psi$ is non-zero on vertices of the graph.

  Then $\boldsymbol{\alpha} = (0,\ldots, 0)$ is a non-degenerate
  critical point of the function $\lambda_n(\boldsymbol{\alpha}) :=
  \lambda_n(\Gamma^{\boldsymbol{\alpha}})$ and its Morse index is
  equal to the \term{nodal surplus} $\phi - (n - 1)$, where $\phi$ is
  the number of internal zeros of $\psi$ on $\Gamma$.

\end{theorem}

As a corollary we get a simple but useful bound on the number of zeros
of $n$-th eigenfunction.

\begin{corollary}
  \label{cor:nodal_bound}
  Let $\Gamma$ be a quantum graph with a Schr\"odinger operator,
  $\lambda_n$ be its $n$-th eigenvalue and $\psi_n$ the corresponding
  eigenfunction.  If $\lambda_n$ is simple and $\psi_n$ does not
  vanish on vertices, the number of zeros $\phi_n$ of the function
  $\psi_n$ is a well-defined quantity which satisfies
  \begin{equation}
    \label{eq:nodal_bound}
    0 \leq \phi_n - (n-1) \leq \beta = |E| - |V| + 1.
  \end{equation}
\end{corollary}

\begin{remark}
  The result of Corollary~\ref{cor:nodal_bound} actually predates
  Theorem~\ref{thm:main_mag} by a considerable time.  It goes back to
  the results on trees \cite{PokPryObe_mz96,Sch_wrcm06}, their
  extension to $\beta>0$ for nodal domains \cite{Ber_cmp08} and to
  number of zeros \cite{BanBerSmi_ahp12}.  But it also follows easily
  from Theorem~\ref{thm:main_mag} due to the Morse index being an
  integer between 0 and the dimension of the space of parameters.
\end{remark}

\begin{exercise}
  \label{hw:even_zeros}
  Let $\Gamma$ be a quantum graph with a Schr\"odinger operator,
  $\lambda$ be its simple eigenvalue and $\psi$ the corresponding
  eigenfunction which does not vanish on vertices.  Prove that $\psi$
  has an even number of zeros on every cycle of the graph $\Gamma$.
\end{exercise}

\subsection{Nodal count of the dihedral graph}
\label{sec:nodal_dihedral}

It turns out there is an explicit formula for the number of zeros of
the $n$-th eigenfunction of the dihedral graph.  This formula was
discovered by Aronovitch, Band, Oren and Smilansky
\cite{BanOreSmi_incoll08}.  The discovery was remarkable as no
explicit formula for the \emph{eigenvalues} of the dihedral graph is
known.

The formula was proved in \cite{BanBerSmi_ahp12} using a fairly
involved construction of opening the graph by attaching two
phase-synchronized infinite leads.  It would be nice to be able to
prove this result from Theorem~\ref{thm:main_mag} directly, but we
know of no such proof.  Instead, we will give here a relatively simple
proof showcasing the power of the interlacing results of
Section~\ref{lem:interlace_ND}, Corollary~\ref{cor:nodal_bound} and
the simple observation that the number of zeros on any cycle of the
graph must be even.  This proof has not previously appeared in any
other source.

\begin{theorem}[Conjectured in \cite{BanOreSmi_incoll08}, first proved
  in \cite{BanBerSmi_ahp12}]
  \label{thm:nodal_dihedral}
  Let $n$-th eigenvalue of a dihedral graph be simple and the
  corresponding eigenfunction not vanish at the vertices (except at the
  Dirichlet vertex).  Then the number of zeros of the eigenfunction is
  \begin{equation}
    \label{eq:nodal_dihedral}
    \phi_n = n - \mod_2\left( \left\lfloor \frac{b+c}{a+b+c}n
      \right\rfloor \right).
  \end{equation}
\end{theorem}

\begin{remark}
  It can be shown that the hypothesis of the theorem is satisfied for
  all eigenvalues and eigenfunction for a generic choice of lengths
  $a$, $b$ and $c$ \cite{BerLiu_jmaa17}.

  It is interesting to observe that the sequence of nodal counts
  $\{\phi_n\}$ contains the information about the relative length of
  the central loop in the graph.
\end{remark}

Before we can prove the theorem, we need two auxiliary lemmas.

\newcommand{\tlambda}{{\tilde\lambda}}
\newcommand{\ts}{{\tilde{s}}}

\begin{lemma}
  Let $\lambda_n$ denote the eigenvalue of the dihedral graph with
  parameters $a$, $b$ and $c$ and let $\tlambda_n$ denote the ordered numbers
  from the set
  \begin{equation}
    \label{eq:tsigma}
    \tilde{\sigma} := \left\{ \frac{\pi}{2a} n_1 \right\}_{n_1\in \N}
    \cup \left\{ \frac{\pi}{2(b+c)} n_2 \right\}_{n_2\in \N}.
  \end{equation}
  Then 
  \begin{equation}
    \label{eq:interlace_dihedral}
    \lambda_n \leq \tlambda_n \leq \lambda_{n+1}, \qquad n=1,2,\ldots 
  \end{equation}
\end{lemma}

\begin{proof}
  Denote by $\sigma$ the spectrum of the dihedral graph.  Consider the
  following sequence of modifications of the graph $\Gamma$ consisting
  of two copies of the dihedral graph, Fig.~\ref{fig:dihedral_mod}(a).
  First we impose the Dirichlet condition on the left attachment point
  of one copy and the right attachment point of the other copy,
  obtaining the graph $\hat\Gamma$, see
  Fig.~\ref{fig:dihedral_mod}(b).  Second we separate the right edge
  of the first dihedral graph and the left edge of the second dihedral
  graph, imposing Neumann conditions on the newly formed vertices, see
  Fig.~\ref{fig:dihedral_mod}(c).  The first modification is covered
  by two applications of Lemma~\ref{lem:interlace_ND}, while the
  second by two applications of Lemma~\ref{lem:interlace_gluing} (in
  reverse).

  \begin{figure}
    \centering
    \includegraphics[scale=0.75]{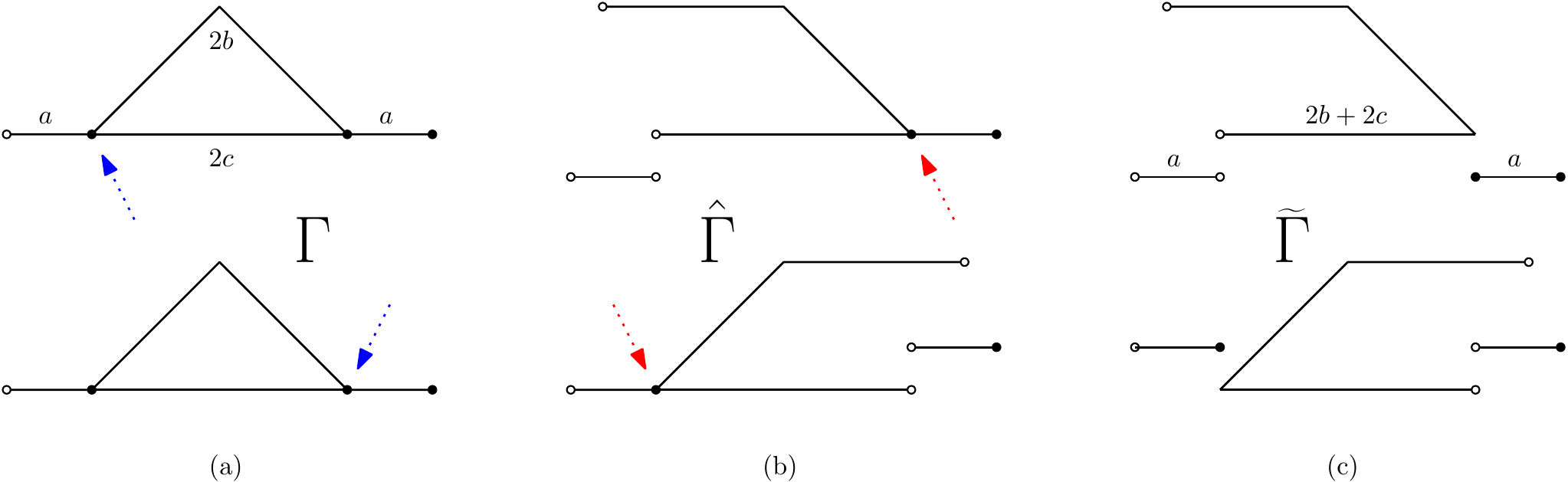}
    \caption{Modifications of the graph $\Gamma$, consisting of two
      copies of the dihedral graph, leading to interlacing
      inequalities \eqref{eq:dihedral_inerlacing1} and
      \eqref{eq:dihedral_interlacing2}.  The Dirichlet vertices are
      distinguished as empty cicrles.}
    \label{fig:dihedral_mod}
  \end{figure}

  Observe that the spectrum of the graph $\Gamma$ is
  $\sigma \cup \sigma$ (in the sense of multisets) and the spectrum of
  the final graph $\widetilde\Gamma$ is
  $\{0\} \cup \tilde{\sigma} \cup \tilde{\sigma}$.  We will denote the
  eigenvalues of the three stages by $s_n$, $\hat{s}_n$ and
  $\tilde{s}_n$ correspondingly.  Because of the degeneracies in the
  spectrum and the interlacing Lemmas, we have the following
  inequalities,
  \begin{equation}
    \label{eq:dihedral_inerlacing1}
    0 < s_2 = s_1 \leq \hat{s}_1 \leq \hat{s}_2 \leq s_4 = s_3 \leq
    \hat{s}_3 \leq \hat{s}_4 \leq \ldots
  \end{equation}
  and
  \begin{equation}
    \label{eq:dihedral_interlacing2}
    \ts_1 = 0 < \hat{s}_1 \leq \ts_3 = \ts_2 \leq \hat{s}_2 
    \leq \hat{s}_3 \leq \ts_5 = \ts_4 \leq \hat{s}_4 \leq \ldots
  \end{equation}
  Combining the two, we obtain
  \begin{equation*}
    \ts_1 = 0 < s_2 = s_1 \leq \ts_3 = \ts_2 \leq s_4 = s_3 
    \leq \ts_5 = \ts_4 \leq \ldots
  \end{equation*}
  Now the claim of the Lemma follows by observing that $s_1 =
  \lambda_2$, $\ts_2 = \tlambda_1$, $s_4 = \lambda_2$, $\ts_4 = \tlambda_2$ etc.
\end{proof}

\begin{lemma}
  \label{lem:students_lemma}
  Let $S_1 = \left\{ \frac{n_1}{\alpha} \right\}_{n_1\in\N}$ and $S_2
  = \left\{ \frac{n_2}{\beta} \right\}_{n_2\in\N}$ with positive
  $\alpha$ and $\beta$.  If
  \begin{equation}
    \label{eq:students_lemma}
    \# \left\{ S_1 \cup S_2 \leq \lambda \right\} = n-1
    \qquad \mbox{then} \qquad
    \# \left\{ S_1 \leq \lambda \right\}
    = \left\lfloor\frac{\alpha}{\alpha+\beta}n\right\rfloor.
  \end{equation}
\end{lemma}

\begin{remark}
  This Lemma was offered as an exercise to the attendees of 2015
  S\'eminaire de Math\'ematiques Sup\'erieures ``Geometric and
  Computational Spectral Theory'' at CRM, University of Montreal.  At
  that time, the author did not know a good proof.

  Several students, including Luc Petiard, Arseny Rayko, Lise Turner
  and Saskia Vo\ss{} submitted proofs.  The proof below is based on
  Arseny's proof with elements borrowed from other attendees' versions.
\end{remark}

\begin{proof}
  Let
  \begin{equation}
    \label{eq:Ndef}
    \# \left\{ S_1 \leq \lambda \right\} 
    = \lfloor \lambda \alpha \rfloor =: N_1, 
    \qquad 
    \#\left\{ S_2 \leq \lambda \right\} 
    = \lfloor \lambda \beta \rfloor =: N_2,
  \end{equation}
  then $N_1 + N_2 = n-1$ and also
  \begin{align*}
    N_1 \leq &\ \lambda \alpha < N_1 + 1, \\
    N_2 \leq &\ \lambda \beta < N_2 + 1.
  \end{align*}
  Multiplying the first inequality by $\beta$ and the second by
  $\alpha$ and going through the middle term, we get
  \begin{equation*}
    \beta N_1 < \alpha (N_2+1), \qquad \alpha N_2 < \beta (N_1+1).
  \end{equation*}
  Adding $\alpha N_1$ to the first and $\alpha(N_1+1)$ to the second
  inequalities we get
  \begin{equation*}
    (\alpha+\beta) N_1 < \alpha (N_1+N_2+1) < (\alpha+\beta)(N_1+1).
  \end{equation*}
  Dividing by $\alpha+\beta$ and using $N_1+N_2+1 = n$ we get
  \begin{equation*}
    N_1 < \frac{\alpha}{\alpha+\beta}n < N_1+1,
  \end{equation*}
  which implies the desired result.
\end{proof}

\begin{proof}[Proof of Theorem~\ref{thm:nodal_dihedral}]
  The eigenvalue $\lambda_n$ of the dihedral graph has $n-1$
  eigenvalues from $\tilde{\sigma}$ below it, so by
  Lemma~\ref{lem:students_lemma}, applied with $\alpha = 2a/\pi$ and
  $\beta = 2(b+c)/\pi$ it has $N_1 = \left\lfloor \frac{a}{a+b+c}n
  \right\rfloor$ eigenvalues from the sequence
  $\left\{\frac{\pi}{2a}\right\}$, which are precisely the eigenvalues
  of the side edges of the dihedral graph but with the Dirichlet
  conditions at the attachment points (i.e.\ the eigenvalues of the
  signle edge components of the graph $\hat\Gamma$).

  By Sturm's theorem, the eigenfunction of the dihedral graph
  corresponding to $\lambda_n$ has $N_1$ zeros on the side edges.
  Therefore $\phi_n - N_1$ zeros lie on the cycle and this quantity
  must be even.  But we know from Corollary~\ref{cor:nodal_bound} that
  $\phi_n$ is either $n-1$ or $n$.  To make $\phi_n - N_1$ even we
  must choose
  \begin{equation*}
    \phi_n = n - \mod_2\left(n-N_1\right) 
    = n - \mod_2\left( \left\lfloor \frac{b+c}{a+b+c}n
      \right\rfloor \right).
  \end{equation*}
\end{proof}

\subsection{Nodal count of a mandarin graph}
\label{sec:nodal_mandarin}

When the bounds of equation (\ref{eq:nodal_bound}) were discovered,
the following question arose: for an arbitrary graph, do all allowed
numbers in the range $0, 1, \ldots, \beta$ appear as the nodal surplus
$\phi_n - (n-1)$ of some eigenfunction?  It turns out the answer is
no.

\begin{theorem} [Band--Berkolaiko--Weyand \cite{BanBerWey_jmp15}]
  \label{thm:mandarin_mag_surplus_bounds} 
  Let $\Gamma$ be a mandarin graph of 2 vertices connected by $d$
  edges.  If the eigenvalue number $n>1$, is simple and the
  corresponding eigenfunction does not vanish on vertices, then the
  nodal surplus $\sigma_{n}$ satisfies
  \begin{equation}
    1\leq\sigma_{n}\leq\beta-1.
    \label{eq:mandarin_mag_surplus_bounds}
  \end{equation}
  In particular, the nodal surplus of the 3-mandarin graph is equal to
  1 for all eigenfunctions except the first, whose nodal surplus is
  always 0.
\end{theorem}

It is interesting to combine this observation with the result of
Theorem~\ref{thm:main_mag}.  The Morse index is never equal to 0 (for
$n>1$) or to the dimension $\beta$, therefore the point
$\boldsymbol{\alpha}=0$ is never a minimum or a maximum.  But the
space of all possible magnetic fluxes is a $\beta$-dimensional torus,
which is compact.  Therefore the extrema must be achieved somewhere!

There are other ``standard'' points where the eigenvalue
$\lambda_n(\boldsymbol{\alpha})$ always have a critical point.  It is
relatively straightforward to extend Theorem~\ref{thm:main_mag} to
points $\left(b_1, \ldots, b_\beta\right)$, where each $b_i$ is 0
or $\pi$.  But it is also straightforward to extend
Theorem~\ref{thm:mandarin_mag_surplus_bounds} (see
\cite{BanBerWey_jmp15} for both extensions) which will show that those
critical point are also never extrema.

\begin{figure}
  \centering
  \includegraphics[scale=0.85]{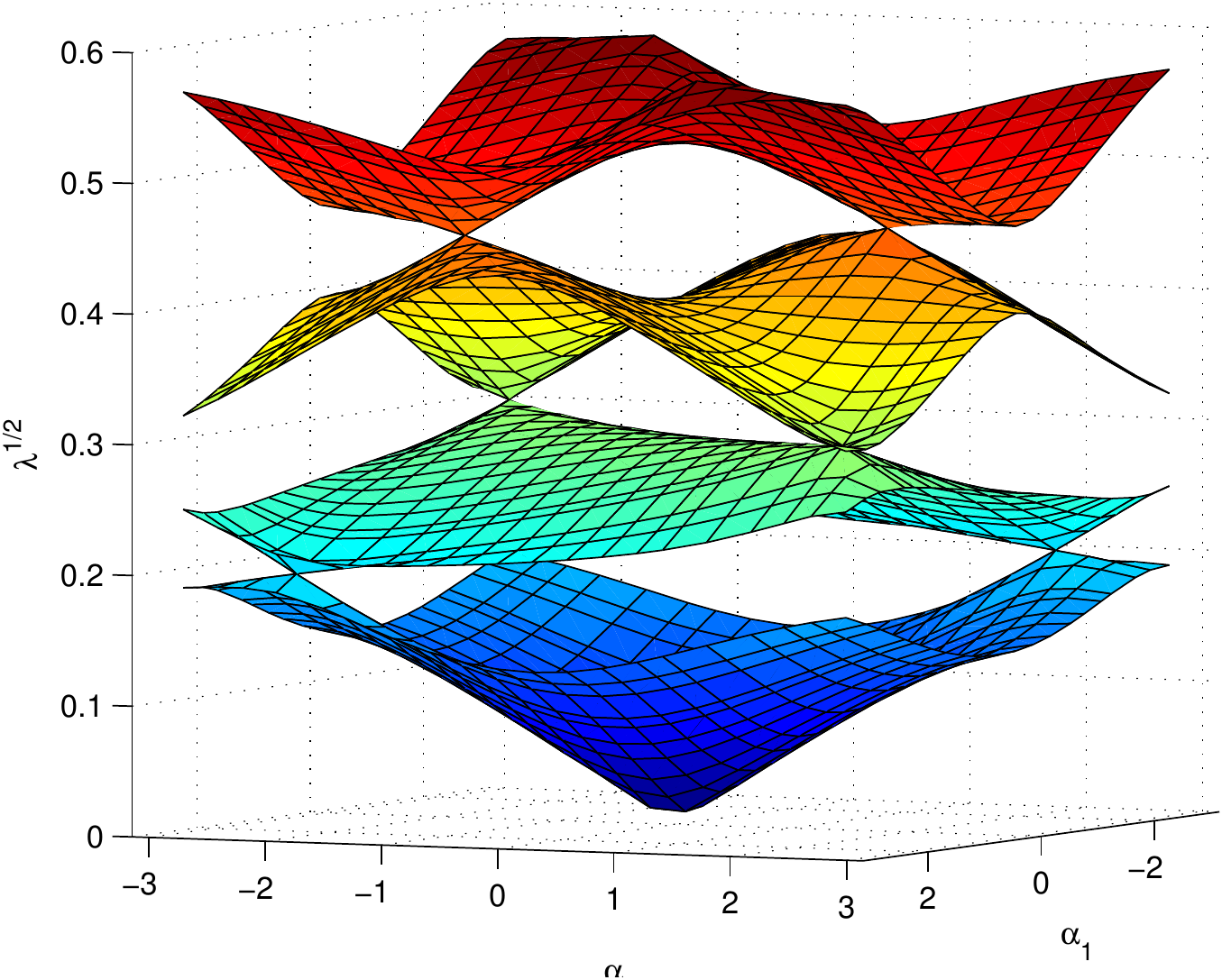}
  \caption{The first four eigenvalues of the mandarin graph as
    functions of two magnetic fluxes $\alpha_1$ and $\alpha_2$.  The
    surfaces can be seen to be touching at conical singularities, the
    so-called Dirac points.}
  \label{fig:f_mandarin_disp}
\end{figure}

The missing extrema turn out to be achieved at singular points of
$\lambda_n(\boldsymbol{\alpha})$, see Fig.~\ref{fig:f_mandarin_disp}
for an example.  Such conical singularities are sometimes called the
Dirac points and their appearance in 3-mandarin graphs is intimately
related to their appearance in the dispersion relation of graphene
\cite{KucPos_cmp07}.  This connection, however, lies outside the scope
of this article.

%
%

%

%

%

\section{Concluding remarks}

There are many interesting topics within the area of spectral theory
of metric graphs that we did not cover.  A very partial list (heavily
biased towards personal preferences of the author) is as follows:

\begin{itemize}
\item Generic properties of eigenfunctions and eigenvalues.  Many
  results, such as Theorem~\ref{thm:main_mag} require the eigenvalue
  to be simple and the eigenfunction to be non-vanishing on vertices.
  These properties are generic with respect to small perturbations of
  the edge lengths (that need to break all symmetries of the graph).
  Results in this direction can be found in \cite{Fri_ijm05},  \cite{CdV_ahp15}(for
  eigenvalues) and \cite{BerLiu_jmaa17} (for both eigenvalues and eigenfunctions).
\item Ergodic flow on the secular manifold.  Barra and Gaspard
  \cite{BarGas_jsp00} introduced an interpretation of the secular
  determinant equation as an ergodic flow piercing a compact manifold
  (more precisely, an algebraic variety) given by solutions of an
  equation like (\ref{eq:Sigma_dihedral}) on a torus.  This
  interpretation leads to many surprising and very general results,
  including those of \cite{BanBer_prl13} and \cite{CdV_ahp15}.
\item Spectral theory of infinite periodic graphs yields a fruitful
  connection to the theory of compact graphs with magnetic field.  The
  background is covered in Chapter 4 of \cite{BerKuc_graphs}; a sample
  of results can be found in \cite{HarKucSob_jpa07,ExnKucWin_jpa10},
  \cite{KucPos_cmp07}, and \cite{BanBer_prl13,BanBerWey_jmp15}.
\item Graphs make a very interesting setting to study resonances.
  Work in this direction has attracted many researchers from across
  the field \cite{KotSmi_prl00}, \cite{TexMon_jpa01},
  \cite{DavPus_apde11}, \cite{KotScha_wrm04}, \cite{GnuSchSmi_prl13},
  \cite{ExnLip_amp07,MR2593999,ExnLip_pla11}.
\item There is an ongoing effort to find bounds on the graph
  eigenvalues (especially the lowest non-zero ones) in terms of the
  geometric properties of the graph, see
  \cite{K2M2_ahp16,BanLev_prep16,Ari_prep16} for the latest results
  and references.
\end{itemize}

\bibliographystyle{plain}
\bibliography{bk_bibl,additional}

\end{document}